\documentclass[11pt]{article}
\usepackage[margin=1in]{geometry}      
\usepackage{lmodern}
\usepackage{tcolorbox}
\usepackage{amsthm,amssymb,amsmath}
\usepackage{amsmath}
\usepackage{tikz,enumitem}
\usetikzlibrary{decorations.pathreplacing}
\tikzset{> = latex}

\usepackage{authblk}

\newcommand{\F }{\textsc{false}}
\newcommand{\Rab }{R_{a\times b}}
\newcommand{\C}{\mathcal{C}}
\newcommand{\maxsatpb}{\texttt{3-Occ-Max-2SAT}}
\newcommand{\minnum}{\texttt{MinNum}}

\newtheorem{theorem}{Theorem}
\newtheorem{definition}{Definition}
\newtheorem{lemma}{Lemma}
\newtheorem{claim}{Claim}
\newtheorem{corollary}{Corollary}
\newtheorem{problem}{Problem}
\newcommand{\T }{\textsc{true}}
\usepackage{graphicx}
\usepackage{amssymb}

\title{Weak  Coverage of a Rectangular Barrier \thanks{Research supported by NSERC, Canada}}

\author[1]{S. Dobrev}
\author[2]{E. Kranakis}
\author[3] {D. Krizanc}
\author[4] {M. Lafond} 
\author[5]{J. Ma\v nuch} 
\author[6]{L. Narayanan}
\author[6]{J.~Opatrny}
\author[7]{S.~Shende}
\author[8]{L. Stacho} 
\affil[1]{Institute of Mathematics, Slovak Academy of Sciences, Bratislava, Slovakia}
\affil[2]{School of Computer Science, Carleton University, Ottawa, Canada}
\affil[3]{Dept. of Mathematics and Computer Science,
Wesleyan University, Middletown CT, USA}
\affil[4]{Department of Mathematics and Statistics, University of Ottawa, Ottawa, Canada}
\affil[5]{Department of Computer Science, University of British Columbia, Vancouver, Canada}
\affil[6]{Department of Computer Science and Software Engineering, 
Concordia University, Montreal, QC,  Canada} 
\affil[7]{Department of Computer Science, Rutgers University, Camden, NJ, USA}
\affil[8]{Department of Mathematics, Simon Fraser University,
Burnaby BC, Canada}

\date{}                                           
\begin{document}
\maketitle
\begin{abstract}
%We consider protection of a 2-dimensional barrier against crossings by a mobile sensor network. A 2-dimensional barrier is represented by a rectangle with sides of similar size. This is a generalization of line barriers and narrow rectangular barriers considered previously in the  literature. 
Assume $n$ wireless mobile sensors are initially dispersed in an ad hoc manner in a rectangular region. They are required to move to final locations so that they can detect any intruder crossing the region in a direction parallel to the sides of the rectangle, and thus provide {\em weak barrier coverage} of the region. We study three optimization problems related to  the movement of sensors to achieve weak barrier coverage: minimizing the {\em number} of sensors moved (MinNum), minimizing the {\em average} distance moved by the  sensors (MinSum), and minimizing the {\em maximum} distance moved by  the sensors (MinMax). We give an $O(n^{3/2})$ time  algorithm for the MinNum problem  for sensors of diameter $1$ that are initially placed at integer positions; in contrast we show that the problem is NP-hard even for sensors of diameter 2 that are initially placed at integer positions. We show that the MinSum problem is solvable in $O(n \log n)$ time for homogeneous range sensors in arbitrary initial positions for the Manhattan metric, while it is NP-hard for heterogeneous sensor ranges for both Euclidean and Manhattan metrics. Finally, we prove that  even very restricted homogeneous versions of the MinMax problem are NP-hard. 

\end{abstract}

\section{Introduction}
Intruder detection is an important application of wireless sensor networks. Each sensor  monitors a circular area centered at its location, and can immediately alert a monitoring station if it detects the presence of an intruder. Collectively the sensors can be deployed  to monitor the entire region, providing so-called {\em area coverage}. However, for many applications, it is sufficient, and much more cost-effective, to simply monitor the boundary of the region, and provide so-called {\em barrier coverage}. 

Barrier coverage was introduced in \cite{kumar2005}, and has been extensively studied since then \cite{bol,ban,mohs,li,yan,adhocnow2009,adhocnow2010}. The problem was posed as the deployment of sensors in a narrow belt-like rectangular region in such a way that any intruder crossing the belt would be detected. A sensor network is said to provide {\em strong barrier coverage} if an
intruder is detected regardless of the path it follows across the given barrier (see Figure \ref{fig:example} (a)). 
In contrast, a sensor network provides  {\em weak coverage} if an intruder is detected when it follows  a straight-line path across the width of the barrier. If the location of the sensors is not known to a trespasser, weak coverage is often sufficient, and is more cost-effective. 

In this paper, we consider a more general notion of weak coverage than previously considered. Given a rectangular barrier, we aim to detect intruders who cross the region in a straight-line path parallel to {\em either of the axes} of the rectangle (see Figure \ref{fig:example} (b)).

\begin{figure}[!htb]
\begin{center}
\includegraphics[width=10cm]{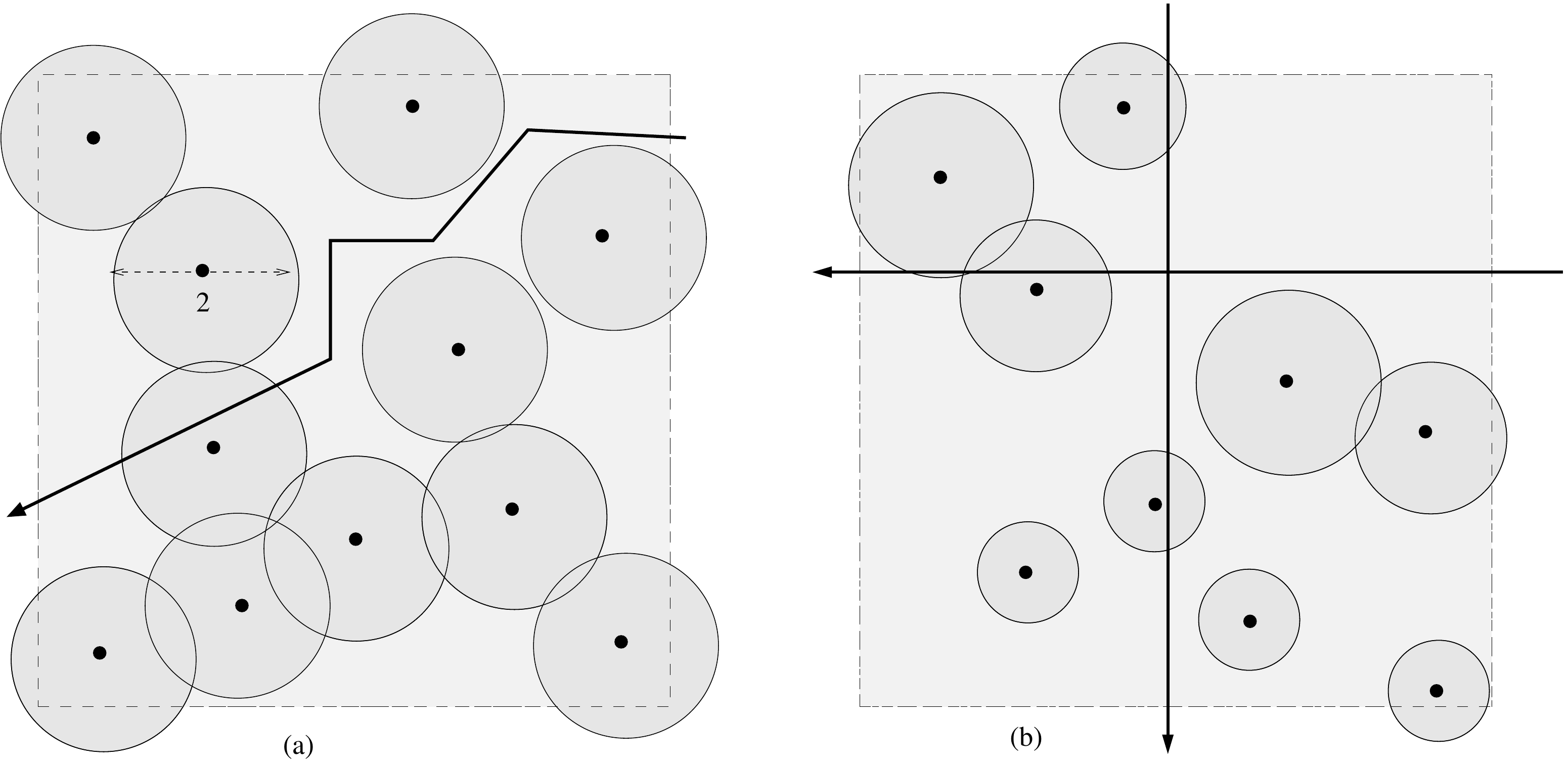}
\end{center}
%\vspace{-20pt}
\caption{(a) Strong  coverage of the shaded square area by a homogeneous network,
  (b) Weak coverage of the shaded square area by a non-homogeneous network for paths perpendicular to the axes.  
}
\label{fig:example}
\end{figure}

A sensor network can be deployed for the given barrier in several different ways. In {\em deterministic deployment}, sensors are placed in pre-defined locations that ensure intruder detection. However, when the deployment area is very large, or the terrain in the area is difficult or dangerous, a deterministic deployment might be costly or even impossible. In those instances a {\em random or ad hoc deployment} of sensors can be done
\cite{bol}. However, this type of deployment might leave some gaps in the coverage of the area. 
Two approaches have been considered in order to deal with this problem. One is a {\em multi-round random deployment} in which the random dispersal is repeated until the coverage of the area is assured with very high probability \cite{yan}.
The other approach is to use {\em mobile (or relocatable) sensors} \cite{tcs2009}. After the initial dispersal,
some or all sensors are instructed to {\em relocate to new locations} so that the
desired barrier  coverage  is achieved. Clearly,  the relocation of the sensors should be performed  in the most efficient way possible. In particular, we may want to minimize the time or energy needed to perform the relocation, or the number of sensors to be relocated. 

%Here two main approaches  have been considered: {\em local, decentralized algorithms} in which a sensor only uses knowledge of positions of sensors in its neighbourhood for its movements \cite{mohs2}, or {\em centralized algorithms}
%in which the movements of each sensor are determined from the complete knowledge of
%positions of all sensors after the dispersal \cite{adhocnow2009,adhocnow2010}. 

%The complexity of the deployment also depends on the type of sensors being used. 

%In centralized algorithms, one is typically interested in optimizing the relocation of sensors with respect to some optimality measures.
%Several optimization measures for sensor relocation have been considered in the literature.
%The relocation problem is called {\em MinSum} if  the new locations minimize the
%sum of movements of the relocated sensors \cite{adhocnow2010}.
%The relocation problem is called {\em MinMax} if the the new locations minimize the
%maximal movement among the sensors \cite{adhocnow2009}, and 
%it is called {\em MinNum} if we minimize the
%number of sensors moved \cite{WCNC2011}.
%Each of these optimization measures is motivated by either minimizing the energy used by the system, or the energy used by any individual sensor, or  minimizing the number of relocation paths needed in the setup of the network.

\subsection{Notation and problem definition}
%In this paper, we study weak coverage of a rectangular barrier $R$ 
%We consider the problem of obtaining a weak crossing coverage of a given rectangular barrier in case when any intruder follows a straight-line path perpendicular to the  rectangle, i.e., as in Figure~\ref{fig:example}~(b).
%The initial deployment is assumed to be adhoc, and sensors are mobile. We consider
%both, heterogeneous and homogeneous  sensor networks, and study the centralized algorithms for getting a weak coverage of the rectangle, starting with the initial adhoc locations of sensors. Each of the three optimization measures stated above is considered. The distance traveled by sensors is 
%measured using the {\em rectilinear (or Manhattan)} metric, i.e., 
%for a pair of points $p_1=[x_1,y_1]$, $p_2=[x_2,y_2]$ in the plane, the distance 
%$d(p_1,p_2):=|x_1-x_2|+|y_1-y_2|$.   
%

We assume that $n$ sensors are initially located in an  axis-parallel rectangular area $R$ of size $a\times b$ in the Cartesian plane.  The  $n$
{\em sensors} $S_1,S_2,\ldots,S_n$ have  sensing ranges
$r_1,r_2,\ldots,r_n$ respectively.  The {\em diameter} of a sensor is equal to twice its range.
We assume that $\sum_{k=1}^n 2r_k \geq \max\{a,b\}$; this ensures that placed in appropriate locations, the sensors can achieve weak barrier coverage.  A sensor network is called {\em homogeneous}  if the sensing ranges of all sensors in the network is the same. Otherwise the network is called {\em heterogeneous}.

A {\em configuration} is a tuple $(R, p_1, p_2, \ldots, p_n)$ where  $R$ is the rectangle to be weakly barrier-covered and $\{p_1,p_2,\ldots,p_n\}$ are the positions of the sensors. We say a configuration is   a {\em blocking configuration} if any straight line, perpendicular to either $x$ or $y$ axes,
crossing  the rectangle $R$, crosses the sensing area of at least one sensor.  In other words, a blocking configuration achieves weak coverage of the rectangle $R$  (abbreviated WCR). A non-blocking configuration is said to have {\em gaps} in  the coverage. A given
configuration $(R,p_1,p_2,\ldots,p_n)$ is said to be an {\em integer configuration} if 
 $p_i=[k_i,j_i]$ for some integers $k_i,j_i$, for every $i$ in the range $1 \leq i \leq n$. We consider both the Euclidean and Manhattan metrics for distance and denote it by $d(x.y)$.

Given an initial configuration $(R, p_1, p_2, \ldots, p_n)$, we study three problems related to finding a blocking configuration: 

%  In the sequel we will abbreviate  the {\em problem of weak crossing coverage of a rectangle} as the {\em WCR problem}.
\begin{itemize}
\item {\bf  MinSum-WCR problem}: Find a blocking configuration $\{R,p_1',p_2',\ldots,p_n'\}$ that minimizes 
$\sum^{n}_{k=1} d(p_k,p_k') $, i.e., the sum of all movements.

\item  {\bf  MinMax-WCR problem:} Find a blocking configuration $\{R,p_1',p_2',\ldots,p_n'\}$ that minimizes 
$\max\{d(p_1,p_1'),d(p_2,p_2'), \ldots, d(p_n,p_n')\} $,
i.e., the size of the maximal move among the sensors.

\item  {\bf MinNum-WCR problem:} Find a blocking configuration $\{R,p_1',p_2',\ldots,p_n'\}$ which
minimizes the number of indices for which $ d(p_k,p_k') \neq 0$,
 $1\leq k \leq n$, i.e., minimizes the number of relocated sensors.

\end{itemize}

%If the given rectangle is very narrow, i.e., more narrow that the sensing diameters of the given sensors,  and the initial position of all sensors is on the horizontal axes, than the WCR problem becomes equivalent to the barrier coverage of a line
%segment. Thus, we should expect MinSum, MinMax and MinNum versions of WCR problems  to be more difficult to solve in general than those for the barrier coverage of
%a line segment. 

\subsection{Our Results}
For the MinNum-WCR  problem, we show that the problem
is NP-complete, even when the  initial configuration is an integer configuration, and even when all sensors have range  $1$.
However  when the initial  configuration is an integer configuration, all sensors have range $0.5$, 
%discrete and all sensors have range $0.5$, 
we give an $O(n^{3/2})$ algorithm for solving the MinNum-WCR problem. 
%It is not sensing areas: it is the "lines they block": let us leave it out for the moment. Observe that in the former case, the initial sensing areas of sensors may overlap partially, while in the latter case, with respect to the vertical and horizontal projections, either the sensing areas have no overlap, or
%they are identical. 
%are initially either co-located or their sensing areas have no overlap. 

When all sensors have the same range, regardless of their initial positions, we give an $O(n\log n)$
algorithm to solve the MinSum-WCR  problem using the Manhattan metric. However, the  problem is  shown to be NP-complete for both Manhattan and Euclidean metrics when the sensors can have different ranges. 

Finally,  we show that the decision version of the MinMax-WCR problem
is NP-complete even for a very restricted case. More specifically,
given an integer configuration, with all sensor ranges equal to $0.5$, 
%discrete configuration where all sensor ranges are  $0.5$, 
the problem of deciding whether there is a blocking configuration with maximal move at
most $1$ (using either the Manhattan or Euclidean metric) is NP-complete. This is in sharp contrast to the one-dimensional barrier coverage case where the MinMax problem can be solved in polynomial time for arbitrary initial positions, and heterogeneous sensor ranges.

\subsection{Related Work}
Barrier coverage using wireless sensors was introduced as a cost-effective alternative to area coverage in  \cite{kumar2005}.  The authors introduced and studied  the notions of  both {\em strong} and {\em weak} barrier coverage in this paper, and studied coverage of a narrow belt-like region. Since then the problem has been extensively studied, for example, see \cite{bol,ban,mohs,li,yan}. 

% Several versions of barrier coverage or area coverage have been studied previously. The complete area coverage was studied for example in \cite{li2011}, and the
%perimeter monitoring of a region was studied in 
%\cite{tcs2009}. 
%The problem of detecting a crossing of a thin rectangular area was introduced in
%\cite{kumar2005}, where the weak and strong coverage is defined. This was also
%subsequently studied in \cite{bol,ban,mohs,li,yan}. 
%The fact that the rectangular area is thin is important in those studies, 
The problem of achieving barrier coverage using mobile or relocatable sensors was introduced in \cite{adhocnow2009}. The authors studied a line segment barrier and gave a polynomial time algorithm for the MinMax problem when all sensors have the same range, and are initially placed on the line containing the barrier.  For the same setting, the case of heterogeneous sensors was shown to be also solvable in polynomial time 
in  \cite{swat2012}, and the algorithm of \cite{adhocnow2009} for the homogeneous case was also improved. An $O(n^2)$ algorithm for the MinSum problem with homogeneous sensors is given  in \cite{adhocnow2010}, and an improved $O(n\log n)$ algorithm is presented in \cite{andr2016}. It was proved in \cite{adhocnow2010} that the MinSum problem is NP-hard when sensors have heterogeneous ranges. The MinNum problem is considered in
\cite{WCNC2011}, and shown to be NP-hard for heterogeneous sensors and poly time  for homogeneous sensors. 

In \cite{tcs2015}, the complexity of the MinMax and MinSum problems when sensors are initially placed in the plane and are required to relocate to cover parallel or perpendicular barriers is studied. The authors show that while MinMax and MinSum can be solved using dynamic programming in polynomial time if sensors are required to move to the closest point on the barrier, even the feasibility of covering two perpendicular barriers is NP-hard to determine.

A stochastic optimization algorithm was considered in \cite{hab2015}.  
Distributed algorithms for the barrier coverage problem were studied in 
\cite{mohs2,distr}. Further, \cite{conf/algosensors/KranakisKLS15} provides algorithms for deciding if a set of sensors provides $k$-fault tolerant protection against rectilinear attacks in both one and two dimensions. To the best of our knowledge, the problem of 
weak coverage of a rectangular region (in two directions) has not been studied previously.

\section{MinNum-WCR Problem}

For a line segment barrier, the MinNum problem was shown to be NP-complete if sensors have different ranges, but  if all sensors have the same range, a polynomial-time MinNum algorithm is given in \cite{WCNC2011}. In this section, we study the MinNum-WCR problem. 

\subsection{Hardness result}

We show that  MinNum-WCR is  NP-complete, even when all sensors have sensing range 1 and the initial configuration is an integer configuration. 
We give a reduction from a restricted satisfiability
problem, shown to be NP-complete in \cite{berman1999some},  and defined below:
%and even APX-complete in~\cite{ausiello2012complexity}.
%It is stated as follows:

%\vspace{3mm}
\begin{tcolorbox}
\noindent{\bf \maxsatpb{} Problem:} \\
\noindent {\bf Input}: An integer $t$, a set of boolean variables $x_1, \ldots, x_n$ and a set of clauses $\C = \{C_1, \ldots, C_m\}$, each consisting of a
conjunction of two literals, such that 
each variable appears in \emph{exactly} $3$ clauses, and no variable occurs only
positively in $\C$, nor only negatively in $\C$. \\
\noindent {\bf Question}: Does there exist an assignment of $x_1, \ldots, x_n$ that satisfies at least $t$ clauses of $\C$? 

\end{tcolorbox}

\begin{theorem}
  MinNum-WCR is NP-complete even for  integer configurations in which all
  sensing diameters are equal to $2$. 
\end{theorem}

\begin{proof}
  Given a \maxsatpb{} instance with variables $x_1, \ldots, x_n$ and clauses
  $\C = \{C_1, \ldots,$ $C_m\}$, we construct a corresponding instance of  the MinNum-WCR problem consisting of a set of sensors $S$ each having radius $r = 1$, and a rectangle $R$ to be covered. 
Note that $m = \frac{3}{2}n$, since there are $3n$ literals in $\C$ and
each clause contains two literals.

$R$ is defined to be a $(6n + 2t)\times (6n + 2t)$ square.
%we set the width and height of the square that needs to be covered
%to $6n + 2t$.
The sensor set $S$ contains one sensor $s_{i, C_j}$ for each literal $x_i$
that appears in a clause $C_j$.  We also need two sensors $\alpha_i$ and $\beta_i$ for each $i \in [n]$. Formally, 
$S = \{s_{i, C_j} : x_i$ occurs in clause $C_j, i \in [n], j \in [m]\} \cup \{\alpha_{1}, \ldots, \alpha_n, \beta_{1}, \ldots, \beta_n\}$.  We first describe how the
sensors of $S$ are laid out on the $y$-axis, then on the $x$-axis.  For a sensor $s \in S$, denote by 
$y(s)$ and $x(s)$ the $y$ and $x$ coordinate
of its center, respectively.
Figure~\ref{fig:minnum_reduction} illustrates the $y$ and $x$ positioning of the sensors.

Each variable $x_i$ has a corresponding gadget $Y_i$ on the vertical axis, where a
gadget is simply a set of sensors positioned in a particular manner.  Each gadget $Y_i$ covers
the vertical range $[6(i - 1) .. 6i]$.
Let
$C_{j_1}, C_{j_2}$ and $C_k$ be the three clauses in which
$x_i$ occurs.  Choose $j_1$ and $j_2$ such that 
$x_i$ occurs in $C_{j_1}$ and $C_{j_2}$ in the same manner
(either positively in both, or negatively in both), 
and so that it occurs in $C_k$ differently.
Let $z = 6(i - 1)$ 
%(which is the highest $y$ coordinate covered by $Y_{i - 1}$ when $i > 1$), 
and let $y(\alpha_i) = z + 1,$ $y(s_{i, C_{j_1}}) = z + 2,$ $y(s_{i, C_k}) = z + 3,$  $y(s_{i, C_{j_2}}) = z + 4$ and $y(\beta_i) = z + 5$.
Observe that $Y_1, \ldots, Y_n$ cover the range $[0 .. 6n]$ on the $y$-axis, which leaves the range 
$(6n .. 6n + 2t]$ uncovered.

Also, note that moving $\alpha_i$ or $\beta_i$ creates a gap in $Y_i$.  Moreover, $s_{i, C_k}$ can be moved, or both $s_{i, C_{j_1}}$ and $s_{i, C_{j_2}}$ can be moved.  However, moving 
both $s_{i, C_k}$ and one of $s_{i, C_{j_1}}$ or $s_{i, C_{j_2}}$ creates a gap (see Figure~\ref{fig:minnum_reduction}).

\begin{figure}[h]
 \centering
  \includegraphics[width=0.95\textwidth]{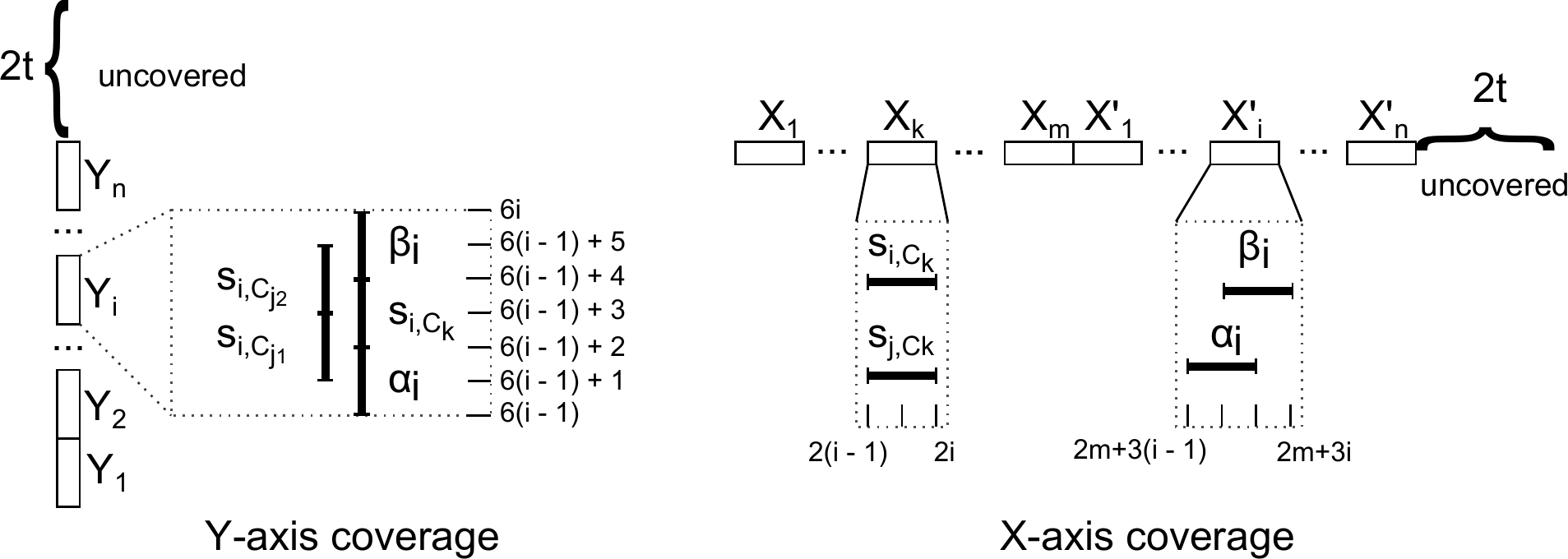}
  \caption{An illustration of the $y$ and $x$ positioning of the sensors in $S$.  The $Y$-axis subfigure only shows the coverage of the sensors on the $Y$ axis, and
    does not depict the $x$ coordinates of the sensors (and the $X$-axis subfigure does not depict the $y$ coordinates).  The depicted $Y_i$ gadget corresponds to
    variable $x_i$ occurring in clauses $C_{j_1}, C_{j_2}$ and $C_k$, where $x_i$
    occurs in the same manner in the first two.  An example of these clauses might be 
    $C_{j_1} = (x_i \vee x_1), C_{j_2} = (x_i \vee x_2)$ and $C_k = (\overline{x_i} \vee x_3)$.  The depicted $X_k$ gadget corresponds to a clause $C_k$ containing the
    two variables $x_i$ and $x_j$.  Finally, the depicted $X'_i$ gadget contains $\alpha_i$ and $\beta_i$, which are partially overlapping.}
  \label{fig:minnum_reduction}
\end{figure}

We now describe how the sensors are laid out on the $x$ axis.
Each clause $C_k \in \C$ has a corresponding gadget $X_k$ that covers the range 
$[2(k - 1) .. 2k]$, and is constructed as follows.
Let $x_i$ and $x_j$ be the two variables occurring in $C_k$.  Then $X_k$ contains the sensors 
$s_{i, C_k}$ and $s_{j, C_k}$, and we set $x(s_{i, C_k}) = x(s_{j, C_k}) = 2k - 1$.
Thus $X_i$ covers the $[2(k - 1) .. 2k]$ range
and $X_1, \ldots, X_m$ cover the range $[0 .. 2m]$.  Note that so far every sensor $s_{i, C_j}$ for $i \in [n], j \in [m]$ has been placed.
We finally create one gadget $X'_i$ for each pair $(\alpha_i, \beta_i)$.  More precisely, for each $i \in [n]$, 
let $X'_i$ contain the $\alpha_{i}, \beta_{i}$ sensors, and set 
$x(\alpha_{i}) = 2m + 3(i - 1) + 1$ 
and $x(\beta_{i}) = 2m + 3(i - 1) + 2$.
Then $X'_i$ covers the range $[2m + 3(i - 1) .. 2m + 3i]$, and 
$X_1, \ldots, X_m, X'_1, \ldots, X'_n$ cover
the range $[0 .. 2m + 3n]$.
Recall that $m = \frac{3}{2}n$, 
and so $2m + 3n = 6n$.
Therefore, the $x$-axis also has the range $(6n .. 6n + 2t]$ uncovered.

It is not hard to see that this construction can be carried out in polynomial time.
We now show that the \maxsatpb{} instance admits an assignment of 
$x_1, \ldots, x_n$ that satisfies $t$ clauses of $\C$
if and only if it is possible to cover the square $R$ by moving $t$ sensors in the corresponding \minnum{} instance.

($\Rightarrow$): Suppose  there is a truth assignment of $x_1, \ldots, x_n$ that satisfies $t$ clauses $C'_1, \ldots, C'_t$ of $\C$.
We claim that it is possible to move $t$ sensors and completely cover $R$.  Observe that the uncovered portion $R'$ of the MinNum-WCR instance consists of a $2t \times 2t$ area. 
For each $k \in [t]$, let $x_{i_k}$ be a variable whose value in the assignment  
satisfies $C'_k$ (either true or false).
Then we move the sensors in the set $S' = \{s_{i_1, C'_1}, s_{i_2, C'_2}, \ldots, s_{i_t, C'_t} \}$, and place them diagonally in $R'$, that is at positions $(6n + 2(i - 1) + 1, 6n + 2(i - 1) + 1)$ for $i \in  \{1, \ldots, t \}$. It is easy to see that $R'$ is now covered. 
To show that the rest of $R$ remains covered, it suffices now to show that no position that was covered before moving $S'$ has become uncovered.  

On the $x$-axis, each sensor of $S'$ belongs to 
a clause gadget $X_i$, and no two sensors of $S'$ belong to the same
clause gadget (since we picked exactly one sensor per
satisfied clause in $S'$).
Since, for each clause gadget $X_i$, there are two sensors
covering the same $x$ range, moving only one cannot leave a portion of $X_i$ uncovered on the $x$ axis.  
We deduce that the $x$-axis is completely covered.
As for the $y$-axis, suppose there is an uncovered position
in some gadget $Y_i$ after moving $S'$.  Let $C_{j_1}, C_{j_2}$ and $C_k$ be the
clauses containing $x_i$, where $x_i$ appears in the same manner in $C_{j_1}$ and $C_{j_2}$.  As $\alpha_i$ and
$\beta_i$ were not moved and yet moving $S'$ leaves uncovered positions within $Y_i$, it must be the case that
$s_{i, C_k}$ and one of $\{s_{i, C_{j_1}}, s_{i, C_{j_2}}\}$, say $s_{i, C_{j_1}}$ without loss of generality,  
were moved.  By the construction of $S'$, this implies that 
the assignment of $x_i$ satisfies clauses $C_k$ and $C_{j_1}$.
However, $x_i$ appears differently in the two clauses 
(positively in one, negatively in the other), 
contradicting the validity of the initial assignment.
We deduce that the $y$-axis is completely covered.

($\Leftarrow$): Suppose that the \minnum{} instance allows 
coverage of the square by moving a set $S'$ of sensors such that 
$|S'| \leq t$.
We claim that the following assignment of $x_1, \ldots, x_n$
is valid and satisfies at least $t$ clauses of $\C$:
for each sensor $s_{i, C_j} \in S' \setminus \{\alpha_1, \ldots, \alpha_n, \beta_1, \ldots, \beta_n\}$, assign to $x_i$ the value that makes it satisfy $C_j$.  If
there are unassigned variables afterward, assign them arbitrarily.  Note that a variable $x_i$ 
might be assigned multiple times, so we must show that it is not assigned both true and false.

First note that as each axis has an uncovered range of length $2t$ and each sensor can only cover a range of $2$ on both axes, we must have $|S'| = t$.  Also note
that, by the same argument, each sensor of $S'$ must be moved to $R'$, that is, the $[6n\ldots 6n+2t]$ range
on both the $x$ and $y$-axes.
In particular, no sensor is moved inside the $6n \times 6n$ area that was
initially covered before moving $S'$.
Thus, the solution
$S'$ cannot contain two sensors 
$s_{i, C_j}$ and $s_{i, C_k}$ such that $x_i$
occurs differently in $C_j$ and $C_k$, as this would create 
a gap in the $Y_i$ gadget on the $y$ axis.  This shows that
$x_i$ cannot be assigned to both true and false, and
so our assignment is valid.
Also, no $\alpha_i$ or $\beta_i$ sensor 
can belong to $S'$, since otherwise a gap would be created either on the $x$ or $y$ axis.
Moreover, $S'$ does not contain two sensors $s_{i, C_k}$ and $s_{j, C_k}$ from the same clause gadget $X_k$, 
as this would create a gap on the $x$ axis.  
We deduce that each sensor of $S'$ belongs to a distinct clause
gadget $X_k$ of $S'$.  Combined with the facts that 
$|S| = t$ and $s_{i, C_j} \in S'$ implies that $x_i$ satisfies $C_j$,
it follows that the constructed assignment satisfies at least $t$ clauses of $\C$.
\end{proof}

\subsection{An efficient algorithm for integer configurations}

We now show that there is a polynomial algorithm to solve the  MinNum-WCR problem for integer  configurations when all sensor %ranges are equal to $0.5$
diameters are equal to 1, and the rectangle to be covered is displaced by 0.5 from the integer grid that contains sensor positions (see Figure~\ref{fig:example1}).  It is not hard to see that there always exists an optimal solution to the MinNum-WCR problem which produces a final blocking configuration which is also an integer configuration. 
%Notice
%that the main difference to the configurations considered in the MinNum NP-completeness proof is the fact that with ranges $0.5$, and sensors in grid positions, any two sensor
%ranges are either disjoint or are completely overlapping with respect to any path perpendicular to the rectangle. 

Consider an integer configuration $(R, p_1, p_2, \ldots, p_n)$ as specified above.  The position of a sensor is a   pair $(i, j)$ where $i$ is said to be the {\em row}   and $j$ is the {\em column} in which the sensor is located. 
A row $i$ (or column $j$ ) so that no sensor is located in it  is called an {\em row gap}  (resp. {\em column gap}), and some sensor needs to be moved to cover such a row or column. Let $r$ and $c$ be the number of rows and columns gaps in the initial configuration. Clearly in the final blocking configuration, there are no uncovered rows or columns. By moving a sensor, we can cover a row or column gap or both. For example, if row $i$ and column $j$ are both gaps, moving a sensor to position $(i, j)$ covers both row $i$ and column $j$. However, moving a sensor may also create a new  row or column gap. To understand better the net effect of moving a sensor based on the other sensors in its row and column, 
%We are interested in finding which sensors are to be relocated so that  we eliminate all gaps in the coverage in $R$ and we minimize the number of sensors move
we introduce the following classification of sensors (see Figure \ref{fig:example1} for an illustration). %To minimize the number of sensors moved, we are interested in finding a maximum-sized subset of free sensors that can be moved without creating additional gaps To find such  a set, we that finds such a set we introduce the following classification of sensors,

\begin{figure}[!htb]
\begin{center}
\includegraphics[width=0.95\textwidth]{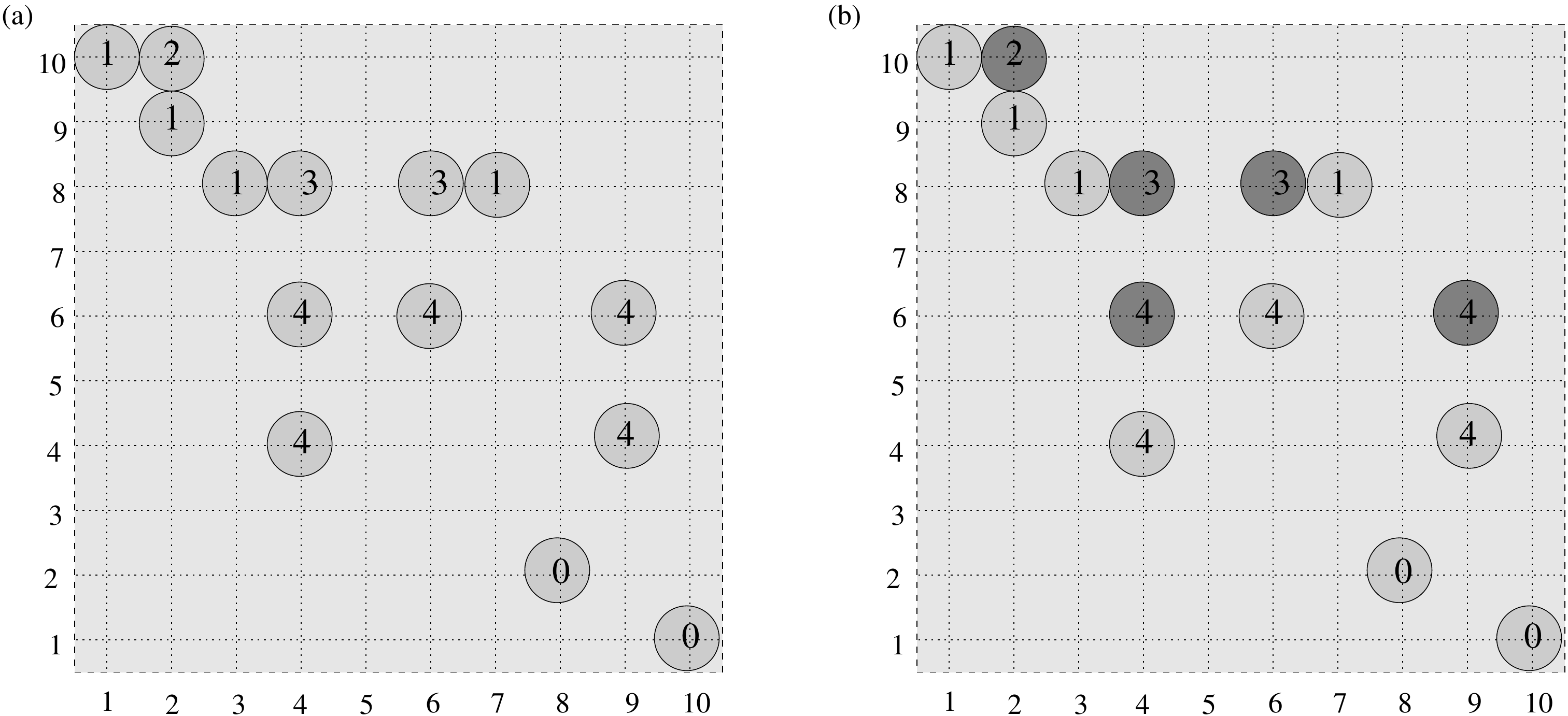}
\end{center}
%\vspace{-20pt}
\caption{(a) Classification of sensors for the MinNum algorithm. Notice that at most 2 free sensors can be removed from row 6 or column 4 without creating
a new gap. (b) A maximum free set of sensors is shown in dark gray.}
\label{fig:example1}
\end{figure}

\begin{definition}
Let $S_k$ be a sensor in position $(i,j)$.
We say that 
\begin{enumerate}
\item  $S_k$  is {\em  free} if there is
at least one other sensor located in row $i$ and at least one other sensor in row $j$.

\item
$S_k$ is of {\em type 0}  if $S_k$ is 
the only sensor located in row $i$ and the only sensor in column $j$.
\item $S_k$ is of {\em type 1}  if either 
there is another sensor located in row $i$ but no other sensor located in column $j$ or there is no other sensor located in row $i$ but there 
is another sensor located in column $j$.
\item
 $S_k$ is of {\em type 2}  if it is a free sensor, and  there is at least one sensor of type $1$ located in row $i$, 
and also at least one sensor  of type 1 in column $j$.
\item
 $S_k$  is of {\em type 3} if it is a free sensor, and  if there is at least one sensor of type  1 located in row $i$ (column $j$)  
and only free sensors in column $j$ (row $j$).
\item
$S_k$  is of {\em type 4} if it is a free sensor, and  the only sensors located in row $i$ and column $j$ are all free  sensors. 
\end{enumerate}
\end{definition}

%\begin{lemma}\label{pigeon}
%  Assume that there is a row gap (respectively column gap) in an integer configuration. Then there is a row (respectively column) containing two sensors, one of which is a sensor of Type 1 or higher.
%\end{lemma}
%\begin{proof} By our assumption, the number of sensors is greater or equal 
%to the size of the grid
%Thus, if there is a row gap, the pigeon hole principle implies that there is a row containing more than one sensor. This sensor is necessarily of Type 1 or higher.  
%  \end{proof}

We call a move of a sensor a {\em sliding move} if the final position of the sensor is either in the same row or column as its initial position. We call a move a {\em jumping} move if the final position is in a different row {\em and} column from its initial position. 

Consider a sensor of type 0. Any  sliding move of such a  sensor  creates an additional row or column gap and can cover at most one other gap.
%previously empty row  or column.
Any jumping move of this sensor  creates both a row and  a column gap,  and it can cover at most one previous row gap and one column gap. 
Thus the total number of row and column gaps cannot decrease  by  moving a sensor of Type 0. In what follows, we can therefore assume that sensors of type 0 and the rows and columns in which they reside are not considered any further. %removed from the input. 
  
A sensor of Type 1 which has another sensor in its row can make a sliding move in its column and cover a row gap.  Any jumping move of this sensor can cover one row and one column gap, but it creates a column gap. Thus  the total number of  row and column gaps decreases by at most 1 by  moving a sensor of Type 1

However, if there is a row gap $i$ and a column gap $j$, then a free sensor (of Type 2, 3, or 4) can make a jumping move to position $[i,j]$ and  cover both row $i$ and column $j$, without creating any new row or column gaps. Therefore, moving a free sensor can reduce the total number of row and column gaps by  $2$. However, moving a free sensor can change the types of sensors in its row or column, and moving {\em multiple} free sensors can create new empty rows or columns as for example removing all free sensors from row  $6$ in Figure \ref{fig:example1}.  

Denote the set of free sensors by $F$. We define a {\em maximum free set} to be a maximum-sized subset $M$ of free sensors that can all be
removed at the same time without creating  new row and column gaps (see Figure~\ref{fig:example1}(b) for an example). The following theorem shows that a maximum free set can be found in polynomial time.

%\begin{tcolorbox}
%%\vspace*{2mm}
%\noindent
%{\bf Algorithm 1} Construction of a {\em Maximum free set of sensors}\\
%\noindent{\bf Input:} A grid configuration. 
%\begin{enumerate}
%\item Determine the type of each sensor.
%\item All sensors of Type 2 are removed and included in the Maximum free set:  
%\item We process Sensors of Type 3 as follows:\\
%Repeat the following:\\
%  Let $S_k$ be a sensor of Type 3. We remove $S_k$ and place it in our Maximum free set.  Recalculate the types of sensors in the row or column from which sensor $S_k$ was removed and which contains no sensor of Type 1.
%\item 
% Finally, process sensors of Type 4 when there are no more sensors of Type 3 as follows:\\
% Let $S_k$ be a sensor of Type 4 in position $[i,j]$.
% Except for sensor $S_K$ put  all sensors of Type 4 in row $i$ and column $j$ in the Maximum free set. This way sensor $S_k$ becomes sensor of Type 0.
% Recompute the types of sensors in rows and columns from which
% a sensor was removed, and repeat the algorithm from Step $2$. 
%\end{enumerate} 
%
%\end{tcolorbox}
%
\begin{theorem}\label{th:MFS}
Given an integer configuration as input,  a maximum free set $M$ can be found in $O(n^{3/2})$ time. 
\label{thm:maxfreeset}
\end{theorem}
\begin{proof}
Define $X$ to be the set of rows and columns that contain only the sensors in $F$, and let $B \subseteq F$ be a minimum-sized subset of $F$ so that every row (or column) in $X$ contains a sensor in $B$. We call $B$ a minimum blocking set for $X$. Then clearly $M$ is a maximum free set if and only if $F-M$ is  a minimum blocking set. 
Therefore, to find a maximum free set, we proceed by finding a minimum blocking set.

Consider a graph $G=(V,E)$ defined as follows. The vertex set $V$ contains a vertex corresponding to  every row and column in $X$; we call the vertex $c_i$ if it corresponds to column $i$ and $r_j$ if it corresponds to row $r_j$. We also introduce two extra vertices $x$ and $y$. For every Type 4 sensor at position $(i,j)$, we introduce an edge $e_{ij}$ between the vertices $r_i$ and $c_j$. For every
Type 3 sensor that has only free sensors in its row $i$ (resp. column $i$), we introduce an edge $e_{ix}$ between vertex $r_i$  (resp. $c_i$) and the vertex $x$. Finally, we add the edge $e_{xy}$ between vertices $x$ and $y$. 

We claim that $B$ is a  blocking set for $X$ if and only if $E' \cup \{ e_{xy} \}$ forms an edge cover in the above graph $G$, where $E'$ is the set of edges corresponding to sensors in $B$. To see this, observe that if a sensor of type $4$ at position $(i,j)$  is in $ B$, then the corresponding edge $e_{ij}$ covers both vertices $r_i$ and $c_j$ in the graph. Similarly, if a type 3 sensor at position $(i, j)$ that has free sensors only in its row (resp.  column) is in $B$ then the corresponding edge $e_{ix}$ covers vertices $r_i$ (resp. $c_i$) and $x$. Furthermore $e_{xy}$ covers both $x$ and $y$. Since $B$ blocks all rows and columns in $X$, it follows that all vertices in $G$ are covered by $E' \cup \{ e_{xy} \}$. Conversely, consider an edge cover in $G$. It must include the edge $\{x, y \}$ since it is the only edge incident on $y$. Additionally, any set of edges that covers the remaining vertices in $G$ must be incident on all vertices corresponding to the set $X$, and therefore corresponds to a set of sensors that blocks $X$. This completes the proof of the claim. 

Since an edge cover can be found via a maximum matching in $O(\sqrt{V} E) = O(n^{3/2})$ time \cite{GareyJohnson}, we can find a minimum blocking set, and thereafter, a maximum free set $M$ in $O(n^{3/2})$ time. \hfill $\Box$
\end{proof}

Assume that in the given MinNum-WCR configuration there are $r$ row gaps and $c$ column gaps. 
We can assume without loss of generality that $r \geq c$.  We now give an algorithm to solve the MinNum-WCR problem: \\
%Thus if our Maximum Free set contains $k$ sensors, $k\geq c$, theup to $c$ empty rows and columns
%could be potentially covered at the same time by $j$ free sensors. This is the essential part of our MinNum algorithm, stated below. \\

\begin{tcolorbox}
\noindent
    {\bf Algorithm 1: The MinNum-WCR Algorithm}\\
{\bf Input}: $I$ an integer configuration with $r$ row gaps and $c$ column gaps, with $r \geq c$ \\

Construct the maximum free set $M$ for $I$.\\
Recalculate the types of sensors after removing the sensors in $M$.\\
Let $k$ be the number of free sensors in $M$.
\begin{description}
\item[Case  $k \geq r$:]  Move $c$ sensors from $M$ using jumping moves to cover  the first $c$ row gaps and all $c$ column gaps. 
Then  $r-c$  sensors from $M$ use sliding moves to cover the remaining row gaps. The total number of sensors moved is $c+(r-c)=r$.
\item[Case  $c\leq k <r$:] Move $c$ sensors from $M$  using jumping moves to cover the first $c$ row gaps and all $c$  column gaps. Move the remaining $k-c$  sensors from $M$  to cover  $k-c$ row gaps.   Finally $r-k$ sensors of Type 1 use sliding moves to cover the remaining row gaps. The total number of sensors moved
  is $c+(k-c)+(r-k)=r$.

\item[Case $k < c$ :] Move $k$ sensors from $M$ using jumping moves to cover $k$ row and column gaps. Then we use sliding moves of sensors of Type 1
to cover the remaining row and column gaps. In total $k +(c-k) + (r-k)=r+c-k$ sensors are moved.
\end{description}
\end{tcolorbox}

\begin{theorem}
  Given an integer configuration, where all sensor ranges are $0.5$, and the rectangle $R$ to be covered is displaced by 0.5 in both axes,  Algorithm 1 solves the MinNum-WCR problem in $O(n^{3/2})$ time. 
\end{theorem}

\begin{proof}
  First we prove %the correctness of
  that Algorithm 1 produces a blocking configuration. Clear\-ly, the sensors in $M$ can be moved without creating new row or column gaps. Assume now that all sensors in $M$ have been used by our algorithm to cover the gaps , but there still remains a row gap (or column gap). Recall that the number of sensors  is assumed to be at least as large as the longer side of the rectangle. 
Thus,  by the pigeonhole principle, there exists a row (or column, respectively)
  that contains more than one sensor. Such sensors must be of Type 1.
One  of them can make a sliding move along its column (or row) to cover the row gap, without creating any column gap or any other row gap. This shows that as long as there are gaps, there are sensors of Type 1 available to fill them as needed in the algorithm. 

  Next we show that Algorithm 1 moves an optimal number of sensors. Given an input configuration with $r$ empty rows and $c$ empty columns, assume
  without loss of generality that $r \geq c$. Observe that at least $r$ sensors need to be moved to cover all empty rows, thus $r$ is a lower bound on the number of sensors to be moved. If $k \geq c$, Algorithm 1 moves exactly $r$ sensors, and it is optimal in the first two cases. 

Suppose instead that  $k < c$. At most $k$ row and column gaps can be covered with sensors from $M$. It follows from the maximality of $M$ that all remaining sensors are  not free, that is, moving any of them must be done using a sliding move, which 
reduces either the number of row gaps or the number of column gaps by at most $1$. Thus, in total we need to move  $k+ (c-k)+(r-k)$
sensors. Therefore Algorithm 2 moves an optimal number of sensors in this case as well.

Given a list of sensors with their coordinates, we can calculate in $O(n)$ time  a list of the number of nodes in each row and a list of the number of nodes in each column. By an $O(n)$  scan of these lists we  can find all nodes of Type $0$ or $1$, and we can mark the rows and columns containing sensors of Type 1.  Now, an additional $O(n)$-time scan of the lists determines the types of all other nodes. By Theorem~\ref{thm:maxfreeset}, a maximum free set of nodes can be constructed in $O(n^{3/2})$ time. After removing the max free set of sensors $M$, we can update the types of nodes in  $O(n)$ time.   Row and column gaps  can be found in  $O(n)$ time. New positions can be calculated in  $O(n)$ time. Thus the total time taken by Algorithm 2 is $O(n^{3/2})$.
\end{proof}

\noindent
{\em Remark:} Notice that the algorithm described in this section assumes (a) that the sensors are initially at integer positions, and (b) that the rectangle $R$ to be covered is displaced from the integer grid by 0.5 in each direction and (c) sensor diameters are 1. All these assumptions are necessary for the algorithm to produce a valid solution.

\section{MinSum-WCR Problem}

In this section, we study the MinSum-WCR problem. The MinSum problem is known to be NP-hard in the one-dimensional case of heterogeneous sensors on a line segment barrier \cite{adhocnow2010}. The following theorem is therefore an immediate consequence: 
\begin{theorem}
  For a heterogeneous sensor network, the MinSum-WCR  problem is NP-complete  
  for both Manhattan and Euclidean metrics.
\end{theorem}

If the sensors are homogeneous and we consider the Manhattan metric, then
to minimize the sum of the movements we can first minimize the sum of all horizontal movements by applying the known, one-dimensional, $O(n\log n)$ MinSum algorithm \cite{adhocnow2010} to the $x$-coordinates of the sensors, and second we minimize the sum of all vertical movements by applying this $O(n\log n)$ MinSum algorithm to  the $y$-coordinates of the sensors. It is easy to see that this gives an optimal
solution. Thus we have the following result.

\begin{theorem}
If all sensor ranges are equal, there is an $O(n\log n)$ algorithm to solve MinSum-WCR for the Manhattan metric.
\end{theorem}

\noindent
{\em Remark:} Unlike the algorithm for MinNum-WCR, the algorithm for MinSum-WCR works for arbitrary input configurations, and sensor ranges that are equal but of arbitrary size.

\section{MinMax-WCR Problem}

As seen in the previous section, the complexity of MinSum-WCR  is very
similar to the complexity of the MinSum barrier coverage of a line segment.
However, this is not the case for the MinMax-WCR problem. For a line segment barrier, the MinMax problem can be solved using an $O(n \log n)$ algorithm even in the heterogeneous case \cite{swat2012}. Surprisingly, the
MinMax-WCR problem is NP-hard even for a integer configuration with a very
restricted possible move size, as shown in Theorem \ref{thm:Iminmax} below.

Throughout this section we will assume that  the rectangle to be covered is displaced by 0.5 in both directions, that the given configurations are integer configurations and all sensor diameters are equal to 1 (i.e., ranges are equal to $0.5$). 
A \emph{line} will be used to represent an interval of length $1$ (either vertical or horizontal),
the line being in the center of the interval.  We place a line at every integer coordinate inside the rectangle on both axes.
We denote the $x$ and $y$  coordinates of the center of a sensor $p$ by $x(p)$ and $y(p)$, respectively.
We then say a sensor $p$ \emph{blocks} a vertical line $x$ if $x(p) = x$, and $p$ blocks a horizontal line 
$y$ if $y(p) = y$ (i.e. $p$ covers the whole interval corresponding to the line).
For convenience, we may define a sensor $p$ as a position, i.e.  we may write $p = (x, y)$
to indicate that $x(p) = x$ and $y(p) = y$.

Notice that for integer configurations,  with sensor ranges equal to $0.5$, if every row and every column of the integer grid
overlaying the given rectangle is blocked by a sensor, then this gives a blocking configuration (in particular, blocking line $1$ on the $x$-axis covers the leftmost part of the rectangle, as it
is displaced by 0.5).
This will be used in the sequel, and in pictures we will only indicate the locations of sensors on the grid. 
First we need some additional definitions.
\begin{definition}
Given configurations $C = (R,p_{1},\dots,p_{k})$ and $C' = (R,q_{1},\dots,q_{k})$, we define \emph{max-$d$-distance} of $C$ and $C'$ as $\max_{1\le i\le k} d(p_{i},q_{i})$. 
\end{definition}

\begin{definition}
  An \emph{$a\times b$ integer rectangle}, denoted by $R_{a\times b}$, is the following subset of
  $\mathbb{Z}^{2}$:
  \begin{equation*}
    \{[i,j]\in \mathbb{Z}^{2}:\ 1\le i\le a,\ 1\le j\le b\} \,.
  \end{equation*}
  %An \emph{integer configuration} $C = (R_{a\times b},p_{1},\dots,p_{k})$ is an integer rectangle $R_{a \times b}$ and $k$ sensors $p_{1},\dots,p_{k}$ in $R_{a \times b}$. 
 A vertical line $x = i$, will be referred to as the \emph{vertical line} $i$.
  Similarly, a horizontal line $y = j$, will be referred to as the \emph{horizontal line} $j$. 
  %We say that point $r$ (respectively, point $s$) is blocking vertical line $i$ (respectively, horizontal line $j$) if $x(r) = i$ (respectively, $y(s) = j$).
  An integer configuration $C$ is called \emph{blocking} if for every integer $1\le i\le a$, vertical line $i$ is blocked by a sensor in $C$ and for every $1\le j\le b$, horizontal line $j$ is blocked by a sensor in $C$.
\end{definition}

%\begin{problem}[Integer Movement MinMax Problem]\label{prob:IMinMax}
%  Given a distance metric $d$, an integer configuration $I = (R,p_{1},\dots,p_{k})$ and a real number $D > 0$, decide if there exits a blocking integer
%  configuration $T = (R,q_{1},\dots,q_{k})$ at max-$d$-distance at most $D$ from $I$.
%\end{problem}

The proof of NP-completeness of the MinMax-WCR problem will follow from the following results. We first show that the \emph{MinMax $(V, H)$-Blocking Problem}, a more general version of MinMax-WCR, is NP-hard.
We then present a reduction from the $(V, H)$-Blocking problem to MinMax-WCR.

\begin{definition}
  Let $R_{a\times b}$ be an integer rectangle, $V$ be a subset of $\{1,\dots,a\} $ and $H$ a subset of $\{1,\dots,b\}$. These sets will represent vertical and
  horizontal lines ($V$-lines and $H$-lines, respectively) that need to be blocked. A configuration $T = (R_{a\times b},q_{1},\dots, q_{k})$ is called \emph{$(V,H)$-blocking} if for every $i\in V$ (respectively, $j\in H$), vertical
  line $i$ (respectively, horizontal line $j$) is blocked in $T$.    
\end{definition}

\begin{problem}[MinMax $(V,H)$-Blocking Problem]\label{prob:IMinMaxVH}
  Given a distance metric $d$, an integer rectangle $R_{a\times b}$, an integer configuration $I = (R_{a\times b},p_{1},\dots,p_{k})$, sets $V\subseteq \{1,\dots, a\} $
  and $H\subseteq \{1,\dots, b\}$, and a real number $D > 0$, decide if there
  exists a $(V,H)$-blocking configuration $T = (R_{a\times b},q_{1},\dots,q_{k})$ at max-$d$-distance at most $D$ from $I$.
\end{problem}

\begin{lemma}\label{lem:IMinMaxVH1}
 The MinMax $(V,H)$-Blocking Problem with Manhattan or Euclidean metric is NP-complete for 
  $D = 1$. In addition, the problem remains NP-complete if the instance
 $I = (R_{a\times b},p_{1},\dots,p_{k})$  with given $V$ and $H$ satisfies: 
 \begin{enumerate}[label = (VH\arabic *)]
 \item \label{item:VH1}
    no sensor of $I$ lies on vertical line $a$
    and on horizontal line $b$,
  \item \label{item:VH2}
    $a\notin V$ and $b\notin H$.
  \end{enumerate}
\end{lemma}
\begin{proof} 
  We will say that a sensor $p$ \emph{partially blocks} a line if $p$ covers a non-empty subset 
of the interval represented by the line, but does not block it.

  Now, for the proof of NP-completeness, it is easy to check that the problem is in NP.
  We will prove the lemma by reduction from 3-SAT(2,2), a SAT problem with 3-clauses only in which every variable has exactly two positive occurrences and two
  negated occurrences. This problem was shown to be NP-complete in \cite{ber2003}. Let $S$ be an instance of 3-SAT(2,2) problem with variables $x_{1},\dots,x_{n}$ and
  clauses $c_{1},\dots,c_{m}$. 
  Start with an integer configuration $I = (\Rab)$ with no sensors, where $a = 16n + 4m$ and $b = 24n$, and $V = H = \emptyset$. 
  For each clause and each variable of the instance, we will include in $I$ a ``gadget'' consisting of a subset of horizontal and vertical lines, some of them
  included in $V$ and $H$, some of them shared with other gadgets and a collection of sensors that lie at intersections of these lines. If in a solution $T$ for configuration $I$, a sensor of a gadget blocks a horizontal or vertical line, we say the
  line is blocked by the gadget. 
Given the initial instance $I = (\Rab, p_1, \ldots, p_k)$ and a solution $T = (\Rab ,q_{1},\dots, q_{k})$, we say that a sensor \emph{moved up} if 
$y(q_i) < y(p_i)$, it \emph{moved down} if $y(q_i) > y(p_i)$, it \emph{moved left} if $x(q_i) < x(p_i)$
and it \emph{moved right} if $x(q_i) > x(q_i)$.

%  Note that since $D = 1$, in any solution $T = (\Rab ,q_{1},\dots, q_{k})$, either $q_{i} = p_{i}$ (we say that $p_{i}$ \emph{stayed}), $q_{i} = p_{i} + (1,0)$ (we say $p_{i}$ \emph{moved right}), $q_{i} = p_{i} + (0,-1)$ (we say $p_{i}$ \emph{moved up}), $q_{i} = p_{i} + (-1,0)$ (we say $p_{i}$ \emph{moved left}), or $q_{i} = p_{i} + (0,1)$ (we say $p_{i}$ \emph{moved down}). 

Construction and properties of gadgets:

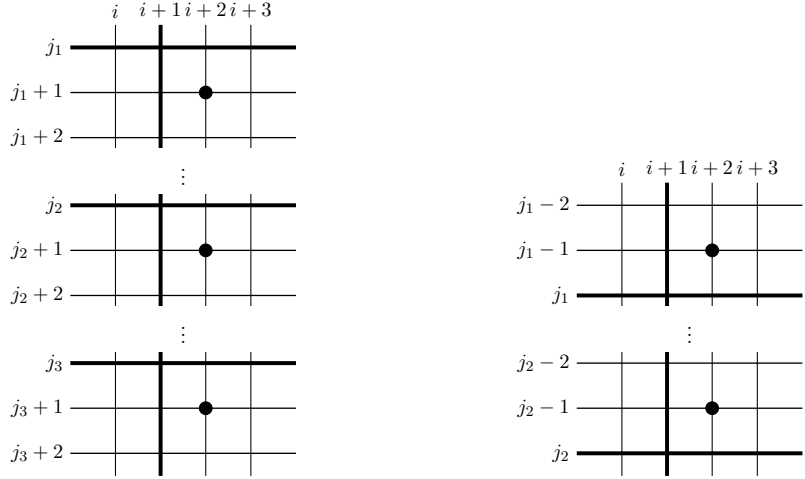
\begin{figure}[h]
  \centering
  \begin{tikzpicture}[scale = .6,every node/.style = {scale = .7}]
    \foreach \y/\lab/\w in {
      1/$j_{3} + 2$/.5,
      2/$j_{3} + 1$/.5,
      3/$j_{3}$/1.5,
      4.5/$j_{2} + 2$/.5,
      5.5/$j_{2} + 1$/.5,
      6.5/$j_{2}$/1.5,
      8/$j_{1} + 2$/.5,
      9/$j_{1} + 1$/.5,
      10/$j_{1}$/1.5}
    {
      \draw [line width = \w pt] (0,\y )--(5,\y );
      \node [left]at (0,\y ) {\lab };
    }
    \foreach \x/\lab in {1/$i$,2/$i + 1$,3/$i + 2$,4/$i + 3$}
    {
      \draw  (\x ,0.5)--(\x ,10.5);
      \node [above] at (\x ,10.5) {\lab };    
    }
    \draw [line width = 1.5pt] (2,0.5)--(2,10.5);
    \draw [fill] (3,2) circle (4pt);
    \draw [fill] (3,5.5) circle (4pt);
    \draw [fill] (3,9) circle (4pt);
    \path [fill = white] (0,3.25) rectangle (6,4.25);
    \node at (2.5,3.75) {$\vdots $};
    \path [fill = white] (0,6.75) rectangle (6,7.75);
    \node at (2.5,7.25) {$\vdots $};
  \end{tikzpicture}
  \hspace{2cm}
  \begin{tikzpicture}[scale = .6,every node/.style = {scale = .7}]
    \foreach \y/\lab/\w in {
      4.5/$j_{2}$/1.5,
      5.5/$j_{2} - 1$/.5,
      6.5/$j_{2} - 2$/.5,
      8/$j_{1}$/1.5,
      9/$j_{1} - 1$/.5,
      10/$j_{1} - 2$/.5}
    {
      \draw [line width = \w pt] (0,\y )--(5,\y );
      \node [left]at (0,\y ) {\lab };
    }
    \foreach \x/\lab in {1/$i$,2/$i + 1$,3/$i + 2$,4/$i + 3$}
    {
      \draw  (\x ,4)--(\x ,10.5);
      \node [above] at (\x ,10.5) {\lab };    
    }
    \draw [line width = 1.5pt] (2,4)--(2,10.5);
    \draw [fill] (3,5.5) circle (4pt);
    \draw [fill] (3,9) circle (4pt);
    \path [fill = white] (0,6.75) rectangle (6,7.75);
    \node at (2.5,7.25) {$\vdots $};
  \end{tikzpicture}  
  \caption{A 3-clause gadget (left) and a 2-clause gadget (right). The $H$-lines and $V$-lines of each gadget are depicted in bold. The sensors of the gadget that are added to configuration $I$ are depicted by dots.
  Note that only the center point of the sensor is represented, and not its whole coverage area.}
  \label{fig:clausegadget}
\end{figure}
 
  \paragraph{3-clause gadgets.} For a clause of $S$ we construct a \emph{3-clause gadget} as follows. The gadget uses 4 consecutive vertical lines, say $i,i + 1,i + 2,i + 3$ and 3 groups of 3 consecutive horizontal lines, say $j_{1},j_{1} + 1,j_{1} + 2$, $j_{2},j_{2} + 1,j_{2} + 2$ and $j_{3},j_{3} + 1,j_{3} + 2$. The vertical lines are not shared with any other gadget, while only horizontal lines shared with other gadgets are $j_{1}$, $j_{2}$ and $j_{3}$ --- each shared with exactly one variable gadget corresponding to one literal of the clause. We add $i + 1$ to $V$, and $j_{1},j_{2},j_{3}$ to $H$. In addition, we add sensors $(i + 2,j_{1} + 1)$, $(i + 2,j_{2} + 1)$ and $(i + 2,j_{3} + 1)$ to $I$, cf. Figure~\ref{fig:clausegadget}. Since $i + 1\in V$ and there are no sensors on vertical lines $i$ and $i + 1$ (they are not shared with any other gadget) and only sensors on vertical line $i + 2$ are the three sensors of this gadget, in any solution $T$, at least one of them has to move left by the full distance $D = 1$. The sensors of the gadget that do not move left by $1$ will correspond to literals that we say are set to value \F{} by $T$, and remaining of the three sensors of the gadget to literals that are set to value \T{} by $T$ (note that setting a negated occurrence $\neg x$ to \T{} corresponds to setting $x$ to \F{}). 

  \begin{claim}\label{cl:clause}
    In any solution $T$ for configuration $I$, for each clause, one of its literals is set to \T{} by $T$.
  \end{claim}

  \begin{proof}
    This follows since at least one sensor of the clause gadget moves left by $1$ to block the $V$-line of the gadget.
  \end{proof}
  
  Note that each of the horizontal lines $j_{1},j_{2},j_{3}\in H$, is either blocked by a sensor of the clause gadget (if the corresponding literal is \F ), or needs to be blocked by a sensor of the corresponding variable gadget, which we now describe. 

\begin{figure}
  \centering
  \begin{tikzpicture}[scale = .6,every node/.style = {scale = .7}]
    \foreach \y/\lab/\w in {
      1/$a_{2} = j_{2} + 1$/1.5,
      2/$a_{4} = j_{2}$/1.5,
      4/$a_{3} = j_{1} + 1$/1.5,
      5/$a_{1} = j_{1}$/1.5}
    {
      \draw [line width = \w pt] (0,\y )--(5,\y );
      \node [left]at (0,\y ) {\lab };
    }
    \foreach \x/\lab in {1/$i$,2/$i + 1$,3/$i + 2$,4/$i + 3$}
    {
      \draw  (\x ,0.5)--(\x ,5.5);
      \node [above] at (\x ,5.5) {\lab };    
    }
    \draw [line width = 1.5pt] (2,0.5)--(2,5.5);
    \draw [fill] (3,1) circle (4pt);
    \draw [fill] (3,4) circle (4pt);
    \path [fill = white] (0,2.5) rectangle (6,3.5);
    \node at (2.5,3) {$\vdots $};
  \end{tikzpicture}  
  \caption{A switch gadget.}
  \label{fig:switchgadget}
\end{figure}
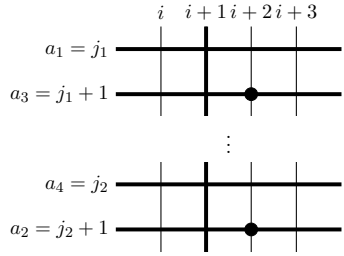

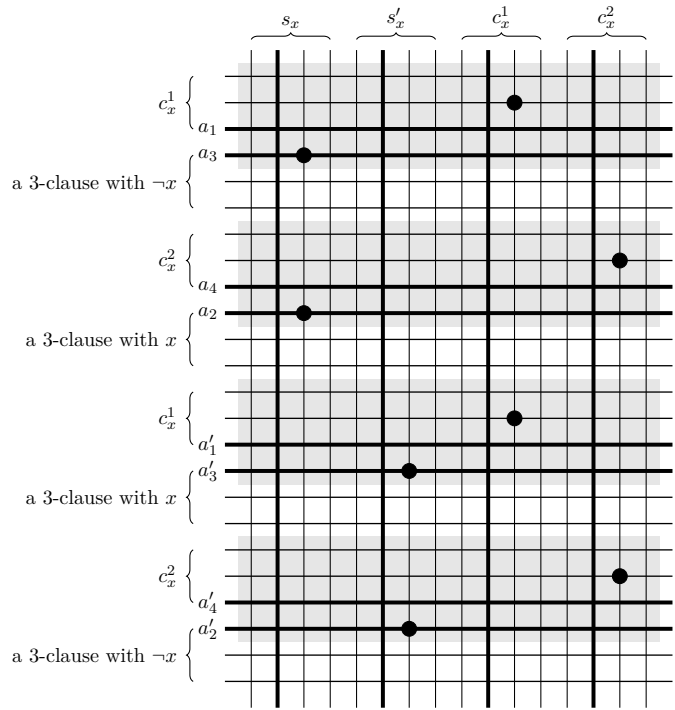
\begin{figure}
  \centering
  \begin{tikzpicture}[scale = .7,every node/.style = {scale = .7}]
    \path [fill = white!90!black] (0.5/2,2.5/2) rectangle (16.5/2,6.5/2);   
    \path [fill = white!90!black] (0.5/2,8.5/2) rectangle (16.5/2,12.5/2);   
    \path [fill = white!90!black] (0.5/2,14.5/2) rectangle (16.5/2,18.5/2);   
    \path [fill = white!90!black] (0.5/2,20.5/2) rectangle (16.5/2,24.5/2);   
    \foreach \y/\lab/\w in {
      24/ /0.5,
      23/ /0.5,
      22/$a_{1}$/1.5,
      21/$a_{3}$/1.5,
      20/ /0.5,
      19/ /0.5,
      18/ /0.5,
      17/ /0.5,
      16/$a_{4}$/1.5,
      15/$a_{2}$/1.5,
      14/ /0.5,
      13/ /0.5,
      12/ /0.5,
      11/ /0.5,
      10/$a_{1}'$/1.5,
      9/$a_{3}'$/1.5,
      8/ /0.5,
      7/ /0.5,
      6/ /0.5,
      5/ /0.5,
      4/$a_{4}'$/1.5,
      3/$a_{2}'$/1.5,
      2/ /0.5,
      1/ /0.5}
    {
      \draw [line width = \w pt] (0,\y/2 )--(8.5,\y/2 );
      \node [left]at (0,\y/2 ) {\lab };
    }
    \foreach \x in {1,...,16} {
      \pgfmathparse{Mod(\x,4)==2?1:0}
      \ifnum \pgfmathresult>0 
      \draw [line width = 1.5pt] (\x/2 ,0)--(\x/2,12.5);      
      \else 
      \draw (\x/2 ,0)--(\x/2,12.5);
      \fi 
    }
    \foreach \x/ \y in {3/15,3/21,7/3,7/9,11/11,11/23,15/5,15/17}
    \draw [fill] (\x/2 ,\y/2 ) circle (4pt);
    \draw [decoration = {brace,raise = 4pt},decorate] (0.5,12.5)--(2,12.5) node [above,midway,yshift = 8pt] {$s_{x}$};
    \draw [decoration = {brace,raise = 4pt},decorate] (2.5,12.5)--(4,12.5) node [above,midway,yshift = 8pt] {$s_{x}'$};
    \draw [decoration = {brace,raise = 4pt},decorate] (4.5,12.5)--(6,12.5) node [above,midway,yshift = 8pt] {$c_{x}^{1}$};
    \draw [decoration = {brace,raise = 4pt},decorate] (6.5,12.5)--(8,12.5) node [above,midway,yshift = 8pt] {$c_{x}^{2}$};
    \foreach \y/ \lab in {11/$c_{x}^{1}$,5/$c_{x}^{1}$,8/$c_{x}^{2}$,2/$c_{x}^{2}$,
      9.5/a 3-clause with $\neg x$,
      6.5/a 3-clause with $x$,
      3.5/a 3-clause with $x$,
      0.5/a 3-clause with $\neg x$
    }
    \draw [decoration = {brace,raise = 12pt},decorate] (0,\y )--(0,\y+1 ) node [left,midway,xshift = -22pt] {\lab };
  \end{tikzpicture}    
  \caption{A variable gadget for $x$ (shaded area) consisting of two switch gadgets $s_{x}$ and $s_{x}'$ and two 2-clause gadgets $c_{x}^{1}$ and $c_{x}^{2}$. Note that some horizontal lines are shared with four 3-clause gadgets: $a_{2}$ and $a_{3}'$ with 3-clause gadgets containing positive occurrences of $x$, and $a_{3}$ and $a_{2}'$ with 3-clause gadgets containing negated occurrences of $x$. Horizontal lines through unshaded areas do not belong to this variable gadget, but belong one of the four 3-clause gadgets.}
  \label{fig:variablegadget}
\end{figure}

  \paragraph{Variable gadgets.} Variable gadgets are constructed using \emph{2-clause gadgets}, which are defined in the same way as 3-clause gadgets but have one less group of horizontal lines and are flipped vertically, cf. Figure~\ref{fig:clausegadget}, and \emph{switch gadgets} described as follows. A switch gadget uses 4 consecutive vertical lines, say $i,i + 1,i + 2,i + 3$ and 2 groups of 2 consecutive horizontal lines, say $a_{1} = j_{1},a_{3} = j_{1} + 1$ and $a_{4} = j_{2},a_{2} = j_{2} + 1$. The vertical lines are not shared with any other gadget. Each of the these horizontal lines is shared with a clause gadget such that $a_{1},a_{2},a_{3},a_{4}\in H$, and $j_{1} - 1,j_{1} + 2,j_{2} - 1,j_{2} + 2\notin H$. We add $i + 1$ to $V$ and sensors $(i + 2,j_{1} + 1)$ and $(i + 2,j_{2} + 1)$ to $I$ (each sensor corresponds to a literal of the clause). As above, at least one of the two sensors has to move left by $D = 1$ in any solution. Consequently, in an integer blocking configuration there are five possibilities which of the horizontal lines $a_{1},\dots,a_{4}$ are blocked by the sensors of the switch gadget: $\{a_{1},a_{2}\}$, $\{a_{2}\}$, $\{a_{2},a_{3}\}$, $\{a_{3}\}$ and $\{a_{3},a_{4}\}$. 

We now introduce a variable gadget for each variable $x$ of $S$. It will consists of two switch gadgets: $s_{x}$ with horizontal lines $a_{1},\dots,a_{4}$ and $s_{x}'$ with lines $a'_{1},\dots,a'_{4}$, and two 2-clause gadgets: $c_{x}^{1}$ with $H$-lines $a_{1}$ and $a'_{1}$ and $c_{x}^{2}$ with $H$-lines $a_{4}$ and $a'_{4}$. To specify the relative position of these four gadgets in $\Rab $, we will denote $i$, $j_{1}$ and $j_{2}$ of gadget $g$ by $i(g)$, $j_{1}(g)$ and $j_{2}(g)$. 
Note that horizontal lines $a_{1},a'_{1},a_{4},a'_{4}$ are shared between the two switch gadgets and two 2-clause gadgets. The remaining horizontal lines of the switch gadgets, $a_{2},a'_{2},a_{3},a'_{3}$, are shared with four 3-clause gadgets for clauses of $S$ containing $x$ or $\neg x$, cf. Figure~\ref{fig:variablegadget}. Note also that we set $i(s_{x}') = i(s_{x}) + 4$, $i(c_{x}^{1}) = i(s_{x}) + 8$, $i(c_{x}^{2}) = i(s_{x}) + 12$, $j_{1}(c_{x}^{2}) = j_{1}(c_{x}^{1}) + 6$, $j_{2}(c_{x}^{1}) = j_{1}(c_{x}^{1}) + 12$ and $j_{2}(c_{x}^{2}) = j_{1}(c_{x}^{1}) + 18$. As a result all four gadgets will be placed into a $16 \times 24$ rectangle as depicted in Figure~\ref{fig:variablegadget}.

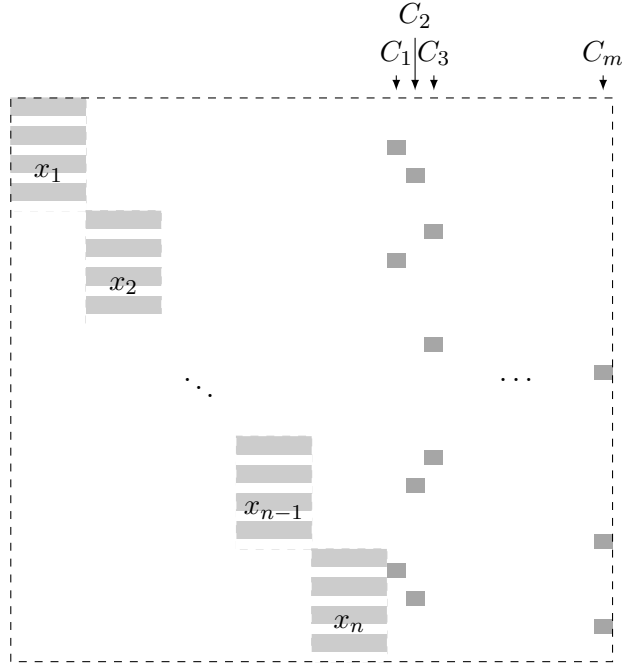
\begin{figure}
  \centering
  \begin{tikzpicture}
    % variable gadget rectangles
    \foreach \x/\lab in {0/$x_{1}$,1/$x_{2}$,3/$x_{n - 1}$,4/$x_{n}$}
    {
      \draw [fill = white!80!black,dashed,color = white!80!black] (\x ,-1.5*\x-1.5) rectangle (\x+1 ,-1.5*\x);
      \foreach \y in {0,1,2,3}
      \path [fill = white] (\x,-1.5*\x-1.5+1.5*\y/4) rectangle (\x+1,-1.5*\x-1.5+1.5*\y/4+0.125);
      \node at (\x+.5,-1.5*\x-1) {\lab };
    }
    \node at  (2.5,-3.75) {$\ddots $};
    % clause gadgets
    \foreach \x/\y/\z in {
      5/0/2,5/1/2,5/4/1,
      5.25/0/3,5.25/3/2,5.25/4/2,
      5.5/1/1,5.5/2/1,5.5/3/1,
      7.75/2/2,7.75/3/4,7.75/4/3
    }
    {
      \path [fill = white!65!black] (\x ,-1.5*\y-0.375*\z ) rectangle (\x+0.25 ,-1.5*\y-0.375*\z+0.1875);
    }
    \node at  (6.75,-3.75) {$\dots $};
    \foreach \x/\lab/\y in {0/$C_{1}$/.3,1/$C_{2}$/.8,2/$C_{3}$/.3,11/$C_{m}$/.3}
    {
      \draw [->] (5.125+\x*0.25,\y )--(5.125+\x*0.25,.1);
      \node [above] at (5.125+\x*0.25 ,\y ) {\lab };
    }
    \draw [dashed] (0,-7.5) rectangle (8,0);
  \end{tikzpicture}
  \caption{An illustration of an integer configuration for a given 3-SAT(2,2) instance containing clause $C_{1} = x_{1}\lor x_{2}\lor \neg x_{n} $ and $C_{m} = x_{1}\lor x_{n - 1}\lor x_{n}$. On the left there are $n$ variable gadgets followed by $m$ clause gadgets. }
  \label{fig:configuration}
\end{figure}

   We have described all ingredients of the construction. To complete the construction we will now describe how the gadgets are placed relatively to each other in $\Rab$. We place the gadget for $x_{1}$ to the top left corner of $\Rab$, the gadget for $x_{2}$ diagonally below, and so on, cf. Figure~\ref{fig:configuration}. Horizontal lines of the clause gadgets are already determined by placement of variable gadgets (see the description of variable gadget), so it is sufficient to choose vertical lines of these gadgets. We place them in the order as they appear in $S$, each taking next the four vertical lines. Note that the construction satisfies that $a\notin V$ and $b\notin H$ and there on sensors on vertical line $a$ and horizontal line $b$, which justifies the second part of the statement of the lemma.

%\textcolor{blue}{EDIT: this is the main Lemma that establishes the fractional/unit movement relationship.  
   %Here I suppose that we have established that a `line' is actually an interval of length $1$, and that blocking the line means covering the interval.
%Also, I use the language of sensors in this proof, since I assume this is how it'll end up.}

Before proceeding to the equivalence between instances, we need the following results.  We first establish the existence of an integer blocking configuration, 
assuming that some solution exists, and show that the variable assignments made by such a configuration must be consistent.

\begin{lemma}\label{lbl:fraction-to-integer}
Suppose that there exists a blocking configuration $T$ for configuration $I$ as constructed above, 
such that each sensor moves by distance at most $D=1$ (under either the Manhattan or Euclidean metric). 
Then $T$ can be transformed, in polynomial time, into an integer blocking configuration $T'$.
\end{lemma}

\begin{proof}
We show how to transform $T$ so that every sensor either stays put, or moves in exactly one direction by $1$ or $-1$. 
We first claim that  $T$ can be transformed into a solution $T'$ in which every sensor moves in only one direction, 
and that the sensors that move horizontally move left by exactly $1$.
First observe that for every sensor $p$, there is a unique vertical line $v(p) \in V$ that $p$ can block (i.e. each sensor has a $V$-line on its left, but not on its right).
Moreover, $v(p)$ is at horizontal distance exactly $1$ from $p$.
Thus for each $v \in V$, there must be at least one sensor $p$ with $v(p) = v$ such that $p$ moves left by $1$ in order to block $v$ (and $p$ does not move vertically).
Suppose now that a sensor $p$ does not block $v(p)$, but $p$ has a non-zero horizontal movement within $(-1, 1]$.
Then there is another sensor $p'$ that blocks $v(p)$, where $p'$ moves left by $1$.   
%If $p$ moves horizontally by some value in $(-1, 0)$ in solution $T$, we may 
%set this movement to $0$  in $T'$ while maintaining coverage of every line (because $p'$ already covers line $v(p)$).
%If instead $p$ moves horizontally by $(0, 1]$ in $T$, then $p$ does not contribute to blocking any line of $V$, even partially (because no sensor
%has a line of $V$ on its right). 
Since $v(p)$ is blocked and $p$ has no $V$-line on its right, we may set the horizontal movement of $p$ to $0$ in $T'$ without affecting the set of blocked lines.  
This proves our claim, since this can be applied to every sensor.
Since this claim holds for either the Manhattan or Euclidean metric, the rest of the proof also applies to both
since no sensor moves diagonally now.

We now assume that the above transformations have been applied and that solution $T$ satisfies our above claim.
We show how each sensor can be made to move by distance $0$ or $1$.
Suppose that there is a sensor that moves by a fractional distance (i.e. in $(0, 1)$).
By our assumption, this sensor must move vertically.
Note that there are $3$ types of sensors: those from 2-clause, 3-clause and switch gadgets.
%We show that in each case, fractional vertical movement can be transformed into unit vertical movement.
We handle each possible case separately.

First let $p$ be a sensor of a 2-clause gadget $c^1_x$ or $c^2_x$.
In $I$, $p$ lies on a line $h \notin H$, with $h + 1, h + 2 \in H$ and $h - 1, h + 3, h + 4 \notin H$.
Moreover there are sensors $q$ on line $h + 2$ (from an $s_x$ or $s'_x$ gadget) and $r$ on line $h + 3$ (from the 3-clause gadget).
Note that there are no other sensors placed on the lines from $h - 1$ to $h + 4$, and so these are the only three available to block
$h + 1, h+2$.
Suppose that $p$ moves vertically by a fractional amount.
If $p$ moves up, then it does not contribute to blocking any line, so its position can be reset and its movement becomes $0$ (in both directions, since $p$ is assumed to move only vertically).
If $p$ moves down, it only partially blocks line $h + 1$, and so sensor $q$ must move up to fill the gap.
But then line $h + 2$ is only partially blocked, and so sensor $r$ must in turn move up.
In other words, all three sensors $p,q,r$ move vertically in solution $T$, and are not used to block any $V$-line.
We can thus obtain an alternate solution $T'$ by resetting the positions of $q$ and $r$, and we move $p$ down by $1$
so that it covers the line $h + 1$.  
Thus we may assume that each sensor $p$ of a 2-clause gadget satisfies our claim.

Suppose now that a sensor $q$ from a $s_x$ or $s'_x$ gadget moves vertically by a fractional amount.
Let $h + 2$ be the line that $q$ lies on in $I$, and define $p$ on line $h$ and $r$ on line $h + 3$ as above.
If $q$ moves up, then $p$ must move down to block the gap on $h + 1$ and $r$ must move up to block the gap on $h + 2$.
Thus all three of $p,q,r$ move vertically, and the alternate solution obtained by not moving $q$ and $r$ and moving $p$
down by $1$ achieves the desired result.
If $q$ moves down, then $p$ must move down by $1$ to block $h + 1$ and $r$ must move up by $1$ to block 
the gap on $h + 2$.  Thus $q$ does not contribute to blocking any line and its position can be reset to $0$ movement.
In this situation all sensors move by a unit amount.
This covers the sensors from $s_x$ or $s'_x$ gadgets.

We now assume that all thee above transformations have been performed, so that sensors from 2-clause or switch
gadget do not move fractionally.
Suppose that a sensor $r$ from a 3-clause gadget lying on line $h + 3$ moves vertically by a fractional amount,
and define the corresponding sensors $p$ and $q$ as above.
If $r$ moves down, it does not contribute to blocking a line and we may reset its position.
If $r$ moves up, then it only partially covers line $h + 2$, and so $q$ must cover the remaining gap.
But as assumed above, $q$ moves vertically by distance either $0$ or $1$, and so in order to block this gap it cannot move vertically at all.
Thus $r$ does not contribute to blocking line $h + 2$, and its position can be reset to $0$ vertical movement.
This concludes the details of the transformation.
The fact that all the necessary modifications can be done in polynomial time is straightforward and omitted.
\end{proof}

Note that for $D = 1$, the movement of a sensor in an integer blocking configuration 
is the same under both the Manhattan or Euclidean metric.  Thus the following results hold in both cases.

%\textcolor{blue}{EDIT: I moved this claim and corollary down, since they need Lemma~\ref{lbl:fraction-to-integer}}
 Recall in the description of 3-clause gadget we set the truth value of each literal. The following claim will help us establish consistency of this assignment for each variable. 
\begin{claim}\label{cl:var}
  For any integer blocking configuration  $T$ for $I$ and any variable $x$ in SAT instance $S$, the sensors of the variable gadget for $x$ block
  either horizontal lines $a_{2},a'_{3}$ and not horizontal lines
  $a_{3},a'_{2}$, or vice versa.
\end{claim}
\begin{proof}
In an integer blocking configuration, the 2-clause gadget $c_{x}^{1}$ can block at most one of the $H$-lines $a_{1}$ and $a_{1}'$. Without loss of generality assume that $a_{1}$ is not blocked by $c_{x}^{1}$. Then it has to be blocked by switch gadget $s_{x}$. Hence, by the discussion above, the set of $H$-lines blocked by $s_{x}$ is $\{a_{1},a_{2}\}$. Therefore, $a_{3}$ and $a_{4}$ are not
  blocked by $s_{x}$. Hence, $a_{4}$
  must be blocked by 2-clause gadget $c_{x}^{2}$, and
  hence, $a'_{4}$ is not blocked by $c_{x}^{2}$. It follows that $s_{x}'$ blocks $a'_{3}$ and $a'_{4}$, and does not
  block $a'_{1}$ and $a'_{2}$. This is the first case of the
  claim. We leave it to the reader to check that assuming $a'_{1}$ is not blocked by $c_{x}^1$ leads to the second case. This concludes the proof of the claim.
\end{proof}

\begin{corollary}\label{col:var}
  For any integer blocking configuration  $T$ for configuration $I$, for each variable $x$, it cannot happen that a positive literal of $x$ is set to \T{} by $T$ and simultaneously, a negative literal of $x$ is set to \T{} by $T$.
\end{corollary}

\begin{proof}
  Assume for contrary that both a positive literal and a negated literal of $x$ are set to \T{} by their 3-clause gadgets. Hence, the $H$-lines corresponding to these literals are not blocked by these 3-clause gadgets, and must be blocked by the variable gadget for $x$, which is not possible by Claim~\ref{cl:var}.
\end{proof}

   \paragraph{Equivalence between instances.}
%\textcolor{blue}{EDIT: changes here and there in this proof to use the integer blocking config obtained from Lemma~\ref{lbl:fraction-to-integer}}

   Let $k = 8n + 3m$. It remains to show that there exists a  $(V,H)$-blocking configuration $T = (\Rab,q_{1},\dots, q_{k})$ at max-$d_{M}$-distance at most $1$ from $I$ if and only if $S$ is satisfiable. Assume that there exists such a configuration $T'$.  Then by Lemma~\ref{lbl:fraction-to-integer}, we can obtain from $T'$ an integer blocking configuration $T$ for $I$.  
We define a variable assignment $\alpha $ as follows: $\alpha (x_{i}) = \T $ if at least one of the positive occurrences of $x_{i}$ in $S$ is set to \T{} by $T$ and $\alpha (x_{i}) = \F $ if both positive occurrences of $x_{i}$ are set to \F{} by $T$. This assignment has the following property: if a literal $\ell $ is set to \T{} by $T$, then $\alpha (\ell ) = \T $ (note that it can happen that $\ell $ is set to \F{} by $T$ and $\alpha (\ell ) = \T $). This is clear for positive occurrences of variables from the definition of $\alpha $. Suppose by contradiction that for a variable $x_{i}$, a negated occurrence $\neg x_{i}$ is set to \T{} by $T$, but $\alpha (\neg x_{i}) = \F $. Then $\alpha (x_{i}) = \T $, and so there is a positive occurrence $x_{i}$ is set to \T{} by $T$, which is a contradiction with Corollary~\ref{col:var}. Now it follows by Claim~\ref{cl:clause} that $\alpha $ is a satisfiable assignment for $S$.

\begin{figure}
  \centering
  \begin{tikzpicture}[scale = .6,every node/.style = {scale = .6}]
    \path [fill = white!90!black] (0.5/2,2.5/2) rectangle (16.5/2,6.5/2);   
    \path [fill = white!90!black] (0.5/2,8.5/2) rectangle (16.5/2,12.5/2);   
    \path [fill = white!90!black] (0.5/2,14.5/2) rectangle (16.5/2,18.5/2);   
    \path [fill = white!90!black] (0.5/2,20.5/2) rectangle (16.5/2,24.5/2);   
    \foreach \y/\lab/\w in {
      24/ /0.5,
      23/ /0.5,
      22/$a_{1}$/1.5,
      21/$a_{3}$/1.5,
      20/ /0.5,
      19/ /0.5,
      18/ /0.5,
      17/ /0.5,
      16/$a_{4}$/1.5,
      15/$a_{2}$/1.5,
      14/ /0.5,
      13/ /0.5,
      12/ /0.5,
      11/ /0.5,
      10/$a_{1}'$/1.5,
      9/$a_{3}'$/1.5,
      8/ /0.5,
      7/ /0.5,
      6/ /0.5,
      5/ /0.5,
      4/$a_{4}'$/1.5,
      3/$a_{2}'$/1.5,
      2/ /0.5,
      1/ /0.5}
    {
      \draw [line width = \w pt] (0,\y/2 )--(8.5,\y/2 );
      \node [left]at (0,\y/2 ) {\lab };
    }
    \foreach \x in {1,...,16} {
      \pgfmathparse{Mod(\x,4)==2?1:0}
      \ifnum \pgfmathresult>0 
      \draw [line width = 1.5pt] (\x/2 ,0)--(\x/2,12.5);      
      \else 
      \draw (\x/2 ,0)--(\x/2,12.5);
      \fi 
    }
    \foreach \x/ \y in {3/15,3/21,7/3,7/9,11/11,11/23,15/5,15/17}
    \draw [fill] (\x/2 ,\y/2 ) circle (4pt);
    \draw [decoration = {brace,raise = 4pt},decorate] (0.5,12.5)--(2,12.5) node [above,midway,yshift = 8pt] {$s_{x}$};
    \draw [decoration = {brace,raise = 4pt},decorate] (2.5,12.5)--(4,12.5) node [above,midway,yshift = 8pt] {$s_{x}'$};
    \draw [decoration = {brace,raise = 4pt},decorate] (4.5,12.5)--(6,12.5) node [above,midway,yshift = 8pt] {$c_{x}^{1}$};
    \draw [decoration = {brace,raise = 4pt},decorate] (6.5,12.5)--(8,12.5) node [above,midway,yshift = 8pt] {$c_{x}^{2}$};
    \foreach \y/ \lab in {11/$c_{x}^{1}$,5/$c_{x}^{1}$,8/$c_{x}^{2}$,2/$c_{x}^{2}$,
      9.5/a 3-clause with $\neg x$,
      6.5/a 3-clause with $x$,
      3.5/a 3-clause with $x$,
      0.5/a 3-clause with $\neg x$
    }
    \draw [decoration = {brace,raise = 12pt},decorate] (0,\y )--(0,\y+1 ) node [left,midway,xshift = -22pt] {\lab };
    \foreach \x/ \y in {3/16,2/21,6/3,7/10,10/11,11/22,15/4,14/17}
    \draw [fill = red] (\x/2 ,\y/2 ) circle (4pt);
  \end{tikzpicture}
  \hfil 
  \begin{tikzpicture}[scale = .6,every node/.style = {scale = .6}]
    \path [fill = white!90!black] (0.5/2,2.5/2) rectangle (16.5/2,6.5/2);   
    \path [fill = white!90!black] (0.5/2,8.5/2) rectangle (16.5/2,12.5/2);   
    \path [fill = white!90!black] (0.5/2,14.5/2) rectangle (16.5/2,18.5/2);   
    \path [fill = white!90!black] (0.5/2,20.5/2) rectangle (16.5/2,24.5/2);   
    \foreach \y/\lab/\w in {
      24/ /0.5,
      23/ /0.5,
      22/$a_{1}$/1.5,
      21/$a_{3}$/1.5,
      20/ /0.5,
      19/ /0.5,
      18/ /0.5,
      17/ /0.5,
      16/$a_{4}$/1.5,
      15/$a_{2}$/1.5,
      14/ /0.5,
      13/ /0.5,
      12/ /0.5,
      11/ /0.5,
      10/$a_{1}'$/1.5,
      9/$a_{3}'$/1.5,
      8/ /0.5,
      7/ /0.5,
      6/ /0.5,
      5/ /0.5,
      4/$a_{4}'$/1.5,
      3/$a_{2}'$/1.5,
      2/ /0.5,
      1/ /0.5}
    {
      \draw [line width = \w pt] (0,\y/2 )--(8.5,\y/2 );
      \node [left]at (0,\y/2 ) {\lab };
    }
    \foreach \x in {1,...,16} {
      \pgfmathparse{Mod(\x,4)==2?1:0}
      \ifnum \pgfmathresult>0 
      \draw [line width = 1.5pt] (\x/2 ,0)--(\x/2,12.5);      
      \else 
      \draw (\x/2 ,0)--(\x/2,12.5);
      \fi 
    }
    \foreach \x/ \y in {3/15,3/21,7/3,7/9,11/11,11/23,15/5,15/17}
    \draw [fill] (\x/2 ,\y/2 ) circle (4pt);
    \draw [decoration = {brace,raise = 4pt},decorate] (0.5,12.5)--(2,12.5) node [above,midway,yshift = 8pt] {$s_{x}$};
    \draw [decoration = {brace,raise = 4pt},decorate] (2.5,12.5)--(4,12.5) node [above,midway,yshift = 8pt] {$s_{x}'$};
    \draw [decoration = {brace,raise = 4pt},decorate] (4.5,12.5)--(6,12.5) node [above,midway,yshift = 8pt] {$c_{x}^{1}$};
    \draw [decoration = {brace,raise = 4pt},decorate] (6.5,12.5)--(8,12.5) node [above,midway,yshift = 8pt] {$c_{x}^{2}$};
    \foreach \y/ \lab in {11/$c_{x}^{1}$,5/$c_{x}^{1}$,8/$c_{x}^{2}$,2/$c_{x}^{2}$,
      9.5/a 3-clause with $\neg x$,
      6.5/a 3-clause with $x$,
      3.5/a 3-clause with $x$,
      0.5/a 3-clause with $\neg x$
    }
    \draw [decoration = {brace,raise = 12pt},decorate] (0,\y )--(0,\y+1 ) node [left,midway,xshift = -22pt] {\lab };
    \foreach \x/ \y in {2/15,3/22,7/4,6/9,11/10,10/23,14/5,15/16}
    \draw [fill = red] (\x/2 ,\y/2 ) circle (4pt);
  \end{tikzpicture}
  \caption{A part of configuration $T$ depicted by red dots: on the left the case when $\alpha (x) = \F $ and on the right the case when $\alpha (x) = \T $.}
  \label{fig:variablesol}
\end{figure}
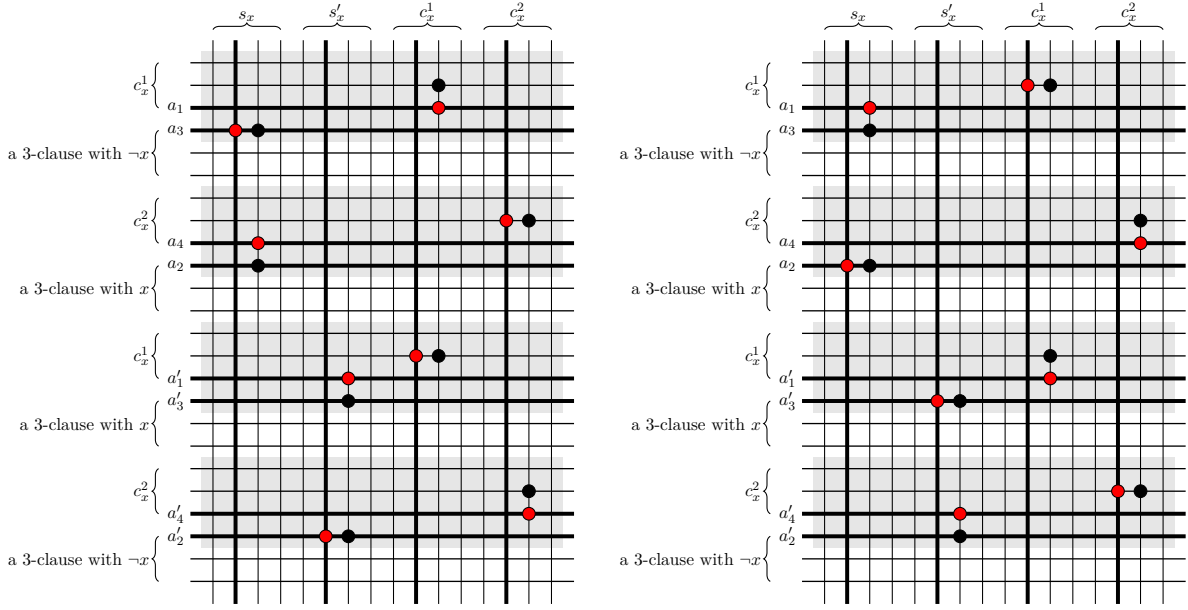

   Next, assume that there exists a satisfiable assignment $\alpha $ for $S$. We will construct a $(V,H)$-blocking configuration $T$ at maximum distance $1$ from $I$ by specifying new positions of sensors for each gadget.  As it turns out, $T$ will also be an integer blocking configuration.  For the variable gadget of a variable $x_{i}$, we distinguish two cases depending on value of $\alpha (x_{i})$ as depicted in Figure~\ref{fig:variablesol}. As one can observe, all $V$- and $H$-lines of the gadget are blocked with exception of two $H$-lines $a_{2},a_{3}'$ that correspond to the positive occurrences of $x_{i}$ in case $\alpha (x_{i}) = \F $ (configuration depicted on the left) or two $H$-lines $a_{2}',a_{3}$ that correspond to the negated occurrences of $x_{i}$ in case $\alpha (x_{i}) = \T $ (configuration depicted on the right). Note that $H$-lines corresponding to occurrences of $x_{i}$ with value \T{} are blocked by the variable gadget.

For each clause, pick a literal $\ell $ with $\alpha (\ell ) = \T $. To obtain $T$ move the sensor corresponding to this literal left and move the remaining two sensors up. The first sensor will block the $V$-line of the gadget and the remaining two sensors block the two corresponding $H$-lines. The $H$-line corresponding to $\ell $ is not blocked by the clause gadget, but since $\alpha (\ell ) = \T $, it is blocked by the variable gadget corresponding to $\ell $. Hence, all $H$-lines corresponding to literals are blocked. Therefore, $T$ is $(V,H)$-blocking and clearly at maximum distance at most $1$ from $I$. This shows that given $I$, $V$ and $H$, it is NP-hard to decide if there exists an integer $(V,H)$-blocking configuration $T = (\Rab,q_{1},\dots,q_{k})$ at maximum distance at most $1$ from $I$.  
\end{proof}

We now present the main result of this section:

\begin{theorem}\label{thm:Iminmax}
  The MinMax-WCR Problem with Manhattan or Euclidean metric is
  %in $P$ for any $D < 1$ and is
  NP-complete for maximum distance $D=1$.
\end{theorem}

\begin{proof}
We will show that the MinMax Problem is NP-complete for $D = 1$ by a reduction from the MinMax $(V,H)$-Blocking Problem % with properties \ref{item:VH1} and \ref{item:VH2}
which is NP-complete by Lemma~\ref{lem:IMinMaxVH1}. %and~\ref{lem:IMinMaxVH2+}. Clearly, the problem is in NP.

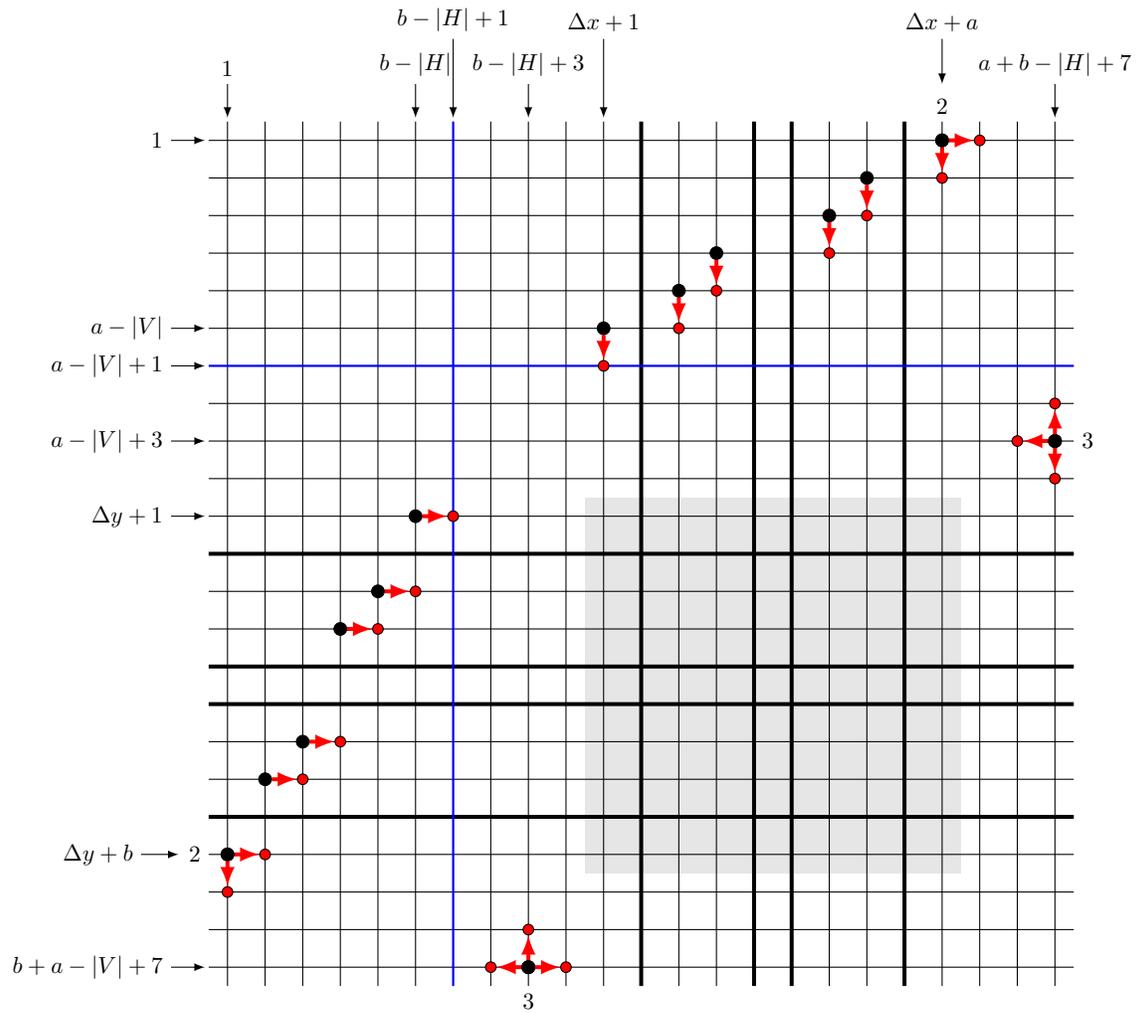
\begin{figure}
  \centering
  
\begin{tikzpicture}[scale = .5,every node/.style = {scale = .8}]
    \path [fill = white!90!black] (9.5,25-12.5) rectangle (19.5,25-22.5);    
    \foreach \x in {0,...,22}
    \draw (\x,-0.5)--(\x,22.5);
    \foreach \y in {0,...,22}
    \draw (-0.5,\y )--(22.5,\y );
    \foreach \x in {11,14,15,18}
    \draw [line width = 1.5pt] (\x,-0.5)--(\x,22.5);
    \foreach \y in {11,14,15,18}
    \draw [line width = 1.5pt] (-0.5,22-\y )--(22.5,22-\y );
    \draw  [line width = 0.8pt,blue] (-0.5,16)--(22.5,16);
    \draw  [line width = 0.8pt,blue] (6,-0.5)--(6,22.5);
    \foreach \x/ \y/ \dx/ \dy in { 
    0/24/1/0,
    0/24/0/1,
    1/22/1/0,
    2/21/1/0,
    3/18/1/0,
    4/17/1/0,
    5/15/1/0,
      10/10/0/1,
      12/9/0/1,
      13/8/0/1,
      16/7/0/1,
      17/6/0/1,
      19/5/0/1,
   %   22/0/1/1,
      19/5/1/0,
      8/27/-1/0,
%      10/27/-1/-1,
      8/27/0/-1,
%      10/27/1/-1,
      8/27/1/0,
      22/13/0/1,
  %    27/10/-1/1,
      22/13/-1/0,
  %    27/10/-1/-1,
      22/13/0/-1
} {
      \draw  (\x, 27-\y) node[draw,circle,minimum size = 6pt,fill = black,inner sep = 0pt] (black) {};
      \draw  (\x+\dx , 27-\y-\dy ) node[draw,circle,minimum size = 5pt,fill = red,inner sep = 0pt] (red) {};
      \draw [line width = 1.5pt,red,->] (black) -- (red);
    }
    \node [left] at (-0.5,27-24) {$2$};
    \node [above] at (19,22+0.5) {$2$};
    \node [below] at (8,22-22.5) {$3$};
    \node [right] at (22.5,22-8) {$3$};
    \foreach \x/\lab/\yb /\yt in {
       0/$1$/0.1/1,
       5/$b - |H|$/0.1/1,
       6/$b - |H| + 1$/0.1/2.2,
       8/$b - |H| + 3$/0.1/1,
       10/$\Delta x + 1$/0.1/2.2,
       19/$\Delta x + a$/1/2.2,
       22/$a + b - |H| + 7$/0.1/1}
    {
      \draw [->] (\x,22.5+\yt )--(\x,22.5+\yb );
      \node [above] at (\x ,22.5+\yt ) {\lab };
    }
    \foreach \x/\lab/\yb /\yt in {
    0/$1$/0.1/1,
    5/$a - |V|$/0.1/1,
    6/$a - |V| + 1$/0.1/1,
    8/$a - |V| + 3$/0.1/1,
    10/$\Delta y + 1$/0.1/1,
    19/$\Delta y + b$/0.8/1.8,
    22/$b + a - |V| + 7$/0.1/1}
    {
      \draw [->] (-0.5-\yt, 22-\x )--(-0.5-\yb, 22-\x );
      \node [left] at (-0.5-\yt, 22-\x ) {\lab };
    }
    \end{tikzpicture}

  \caption{The reduction from $(V,H)$-blocking problem Movement MinMax. A dot with number $c$ depicts $c$ sensors at the same position. The shaded square represents the original instance $I$. As before the bold lines depict lines in $V$ and $H$. The construction guarantees that the sensors around the shaded square have to block all lines not in $V$ and $H$, and cannot block any lines in $V$ and $H$, cf. the proof. Blue lines indicate the starting point of this proof.}
  \label{fig:blocking}
\end{figure}

Let $I = (\Rab ,p_{1},\dots, p_{k}),H,V$ be an instance of the MinMax $(V,H)$-Blocking Problem satisfying properties \ref{item:VH1} and \ref{item:VH2}.
Let $R': = R_{(a + b - |H| + 7)\times (b + a - |V| + 7)}$.
We construct a new instance $$I' = (R',p_{1}',\dots, p_{k}',v_{1},\dots, v_{a - |V| + 4},h_{1},\dots, h_{b - |H| + 4})$$ of the MinMax Problem as follows. We add $\Delta x: = b - |H| + 4$ vertical lines before rectangle $\Rab $ and $3$ vertical lines after $\Rab $. Similarly, we add $\Delta y := a - |V| + 4$ horizontal lines above $\Rab $ and $3$ horizontal lines below. This defines the rectangle $R'$ of $I'$. Consequently, every sensor $p_{i}$ of the original instance $I$ is shifted by $(\Delta x,\Delta y)$, i.e., $p_{i}' = (x(p_{i}) + \Delta x,y(p_{i}) + \Delta y)$. Next, we define sensors $v_{1},\dots, v_{a - |V| + 4}$. Let $x = c$ be the $i$-th vertical line of $\Rab $ not in $V$. We set $v_{i} = (c + \Delta x,a - |V| + 1 - i)$. This defines the first $a - |V|$ sensors of the sequence. Next, add another sensor $v_{a - |V| + 1}$ at the same
position as $v_{a - |V|}$. Finally, set the remaining $3$ sensors of the sequence to $(a + b - |H| + 7,a - |V| + 3)$. The sensors $h_{1},\dots, h_{b - |H| + 3D + 1}$ can be set symmetrically, cf. Figure~\ref{fig:blocking}.

Next we will show that, under both the Manhattan or Euclidean metrics, there exists a configuration $T$ at maximum distance at most $1$ from $I$ that is $(V,H)$-blocking if and only if there exits a configuration $T'$ at maximum distance at most $1$ from $I'$.

  Assume such a $T = (\Rab ,q_{1},\dots, q_{k})$ exists. We construct $$T' = (R',q_{1}',\dots, q_{k}',v_{1}',\dots, v_{a - |V| + 4}',h_{1}',\dots, h_{b - |H| + 4}')$$ as follows:
  \begin{itemize}
  \item to obtain $q_{i}'$ move each sensor $p_{i}'$ in the same direction $p_{i}$
    moved in $T$;
  \item for $i\in \{1,\dots, a - |V|\} $ to obtain $v_{i}'$ move sensor $v_{i}$ $1$ line down;
  \item move sensor $v_{a - |V| + 1}$ $1$ line right, sensor $v_{a - |V| + 2}$ $1$ line left,
    sensor $v_{a - |V| + 3}$ $1$ line up and sensor $v_{a - |V| + 4}$ $1$ line down;
  %\item for $i\in \{0,1, 2\} $, to obtain $v_{a - |V| + 2 + i}'$ move sensor $v_{a - |V| + 2 + i}$ 
  %in direction $(|1 - i| - 1,-1 + i)$ (the new location of these sensors are all at distance $1$ from $v_{a - |V| + 2}$);
  \item to obtain sensors $h_{1}',\dots, h_{b - |H| + 4}'$ move sensors $h_{1},\dots, h_{b - |H| + 4}$ symmetrically to corresponding cases above;
    cf. Figure~\ref{fig:blocking}.
  \end{itemize}
  Clearly, since each sensor in $T$ moved by Manhattan/Euclidean distance at most $1$, each sensor in $T'$ has moved Manhattan/Euclidean distance at most $1$. Therefore, $T'$ is maximum distance at most $1$ from $I'$. 
  Next, we verify that each vertical line of $R'$ is blocked. The vertical lines $1,\dots, b - |H| + 1$ are blocked by sensors $h_{b - |H + 1}',\dots, h_{1}'$. The next $3$ vertical lines are blocked by sensors $h_{b - |H| + 2}',\dots, h_{b - |H| + 4}$. The next $a$ vertical lines (the lines of the original rectangle $\Rab $) are blocked by either $q_{1}',\dots, q_{k}'$ or $v_{1}',\dots, v_{a - |V|}'$ depending on whether the line is in $V$ or not. The remaining $3$ vertical lines are blocked by sensors $v_{a - |V| + 1}',\dots,v_{a - |V| + 4}'$. Similarly, one can argue that every horizontal line is blocked. Hence, $T'$ is blocking, as required.

For the converse, suppose such a $T' = (R',q_{1}',\dots, q_{k}',v_{1}',\dots, v_{a - |V| + 4}',h_{1}',\dots, h_{b - |H| + 4}')$ exists. 
We will first show that sensors $v_{1}',\dots, v_{a - |V| + 4}',h_{1}',\dots, h_{b - |H| + 4}'$ do not block lines in $V$ and $H$.   In other words, the moves illustrated in Figure~\ref{fig:blocking} are forced.  To this end, consider horizontal line $a - |V| + 1$ (see the blue horizontal line in Figure~\ref{fig:blocking}). The only sensor of $I'$ at (Manhattan and Euclidean) distance at most $1$ from this line is $v_{1}$ (note that there are $3$ horizontal lines separating this line from the $h$-sensors and $p'$-sensors). Therefore, in $T'$, $v_{1}'$ has moved $1$ line down. Now, consider horizontal line $a - |V|$. The only sensors at distance at most $1$ are $v_{1}$ and $v_{2}$. Since $v_{1}'$ does not lie on this line, $v_{2}'$ must lie on it, i.e., it has moved $1$ line down. Inductively, it follows that all sensors $v_{1}',\dots, v_{a - |V| - 1}'$ have moved $1$ line down. Now, consider the vertical line $\Delta x + a + 1$. Since the instance $I$ satisfies \ref{item:VH2}, the sensors $v_{a - |V|},v_{a - |V| + 1}$ (lying at the same position) lie on the last vertical line of rectangle $\Rab $. Since $I$ satisfies \ref{item:VH1}, the vertical line $\Delta x + a + 1$ can only be blocked by either $v_{a - |V|}'$ or $v_{a - |V| + D}'$.  This sensor must move at least one line right, therefore, it cannot block the horizontal line $2$. It follows the remaining sensor of $\{v_{a - |V|}', v_{a - |V| + 1}'\}$ must move $1$ line down to block this line. 
We have that none of these sensors in $T'$ considered so far can block lines in $V$. The remaining $v'$-sensors and $h'$-sensors are at distance more than $1$ from $V$-lines, therefore the $V$-lines must be blocked by $q'$-sensors. Symmetrical argument shows that also the $H$-lines must be blocked by $q'$-sensors.

%I removed the resetting because it doesn't really matter that we end up outside...
%and resetting was complicating things when allowing fractional movements
Now, we are ready to construct $T = (\Rab ,q_{1},\dots, q_{k})$. For each $i\in \{1,\dots ,k\} $, obtain $q_{i}$ from $q_{i}'$ by adding vector $(-\Delta x,-\Delta y)$
%and resetting coordinates which ended up outside of $\Rab $. 
%More precisely, if $x(q_{i}') - \Delta x < 1$ or $x(q_{i}') - \Delta x > a$, set $x(q_{i}) := x(p_{i})$, otherwise
i.e. set $x(q_{i}) := x(q_{i}') - \Delta x$
%. Similarly, if $y(q_{i}') - \Delta y < 1$ or $y(q_{i}') - \Delta y > b$, set $y(q_{i}) := y(p_{i})$, otherwise set 
and $y(q_{i}) := y(q_{i}') - \Delta y$. 
Since 
sensors $p_{i}'$ and $q_{i}'$ have distance at most $1$,
%and 
%resetting can only decrease this distance,  
$T$ is at maximum distance at most $1$ from $I$. 
%Now, it suffices to show that $T$ is $(V,H)$-blocking. It is clear that the $x$-coordinate of a sensor $q_{i}'$ 
%blocking a $V$-line in $R'$ has not been reset when defining $q_{i}$. Therefore, this $V$-line in $\Rab $ is 
%blocked by $q_{i}$. Similar argument holds for $H$-lines. Therefore, $T$ is $(V,H)$-blocking. 
It is also easy to see that $T$ is $(V, H)$-blocking, since no sensor outside of $q_1', \ldots, q_k'$ was used to block any line of the MinMax instance corresponding to $V$ and $H$ after shifting (even partially).
Thus $q_1', \ldots, q_k'$ block the shifted $V$-lines and $H$-lines of the instance $I$, 
implying that $q_1, \ldots, q_k$ must be $(V, H)$-blocking.
This completes the 
\end{proof}

\section{Conclusion}
In this paper we studied the complexity of establishing weak barrier coverage (WCR) in a given rectangular area using mobile sensors so that the network can detect any crossing of the area in a direction perpendicular to the sides of the rectangle. We considered the three
typical optimization measures MinSum, MinMax, and MinNum for movements of sensors. 
For the MinNum-WCR problem, we show that the problem is NP-hard if sensors have sensing diameter 2, even if sensors are placed initially at integer locations. On the other hand, if sensors of sensing diameter 1 are placed at integer locations, we show an $O(n^{3/2})$ algorithm to solve the problem. For the MinMax-WCR problem, we show that  the problem is NP-complete for both Euclidean and Manhattan metrics even when sensors are initially placed in integer locations. For the MinSum-WCR problem using the Manhattan metric, we showed that the problem is NP-complete for heterogeneous sensors, and solvable in $O(n \log n)$ time for homogeneous sensors. The complexity of the MinSum-WCR problem for the Euclidean metric remains unknown.

\bibliographystyle{plain}
\bibliography{refs1}
\newpage

\end{document}